\documentclass[12pt]{article}

\usepackage{epsf,epsfig,amsfonts,a4wide}
\usepackage{mathdots}
\usepackage{amsmath,amssymb,amsthm}
\usepackage{url}
 \usepackage{graphicx,color, xcolor}
\usepackage{enumerate}
  \usepackage{mathtools}
\usepackage{dutchcal}

\newtheorem{Theorem}{Theorem}
\newtheorem{Lemma}{Lemma}[section]
\newtheorem{Remark}{Remark}[section]

\newtheorem{Corollary}{Corollary}[section]

\newtheorem{Ex}{Example}[section]

\newcommand{\be}{\begin{equation}}
\newcommand{\ee}{\end{equation}}

\newcommand{\tr}{\mathrm{tr}\,}

\newcommand{\R}{\mathbb{R}}
\newcommand{\Id}{\operatorname{Id}}

\newcommand{\ddd}{\mathrm{d}\,}

\newcommand{\diag}{\operatorname{diag}}

\newcommand{\vpd}[2]{\frac{\delta#1}{\delta#2}}

\newcommand{\pd}[2]{\frac{\partial#1}{\partial#2}}

\newcommand{\weg}[1]{}

\title{Applications of Nijenhuis Geometry IV: multi-component KdV and Camassa-Holm equations}
\author{Alexey V. Bolsinov\footnote{ School of Mathematics,
 Loughborough University,
 LE11 3TU, UK \ \ 
 \quad {\tt A.Bolsinov@lboro.ac.uk} } \quad
\& \quad  Andrey Yu. Konyaev\footnote{Faculty of Mechanics and Mathematics, Moscow State University, 119992, Moscow Russia
 \ \ \quad {\tt  maodzund@yandex.ru}} \quad \& \quad Vladimir S. Matveev\footnote{
Institut f\"ur Mathematik, Friedrich Schiller Universit\"at Jena,
07737 Jena Germany  \ \ \quad {\tt  vladimir.matveev@uni-jena.de}} 
}  
\date{}

\begin{document}

\maketitle

\begin{abstract}

We  construct  a new series of  multi-component integrable PDE systems that contains as particular examples (with appropriately chosen parameters) and generalises many famous integrable systems including KdV,  coupled  KdV \cite{fordy}, Harry Dym, coupled Harry Dym \cite{fordy2}, Camassa-Holm, multi-component Camassa-Holm \cite{ih}, Dullin-Gottwald-Holm and Kaup-Boussinesq systems. The series also contains integrable systems with no low-component analogues.   
 
\end{abstract}

\tableofcontents

\section{Introduction} \label{sect1}

In our paper we  consider  two  types of PDE systems. The first one is an $n$-{\it component evolutionary system}  of PDEs:
\begin{equation}\label{eq:1}
u^i_t = \xi^i(u), \quad i = 1, \dots, n,    
\end{equation}
where $u^i=u^i(x,t)$ are unknown functions and $\xi^i(u)$ is a differential polynomial in $u^1,\dots, u^n$, that is a  polynomial in derivatives $u^j_x, u^j_{xx}, u^j_{xxx}, \dots$, $j = 1, \dots, n$, whose coefficients are functions of $u^1, \dots, u^n$. 
For $n = 1$,  an  example of such a system is the famous  KdV equation
$$
u_t = \frac{1}{2} u_{xxx} + \frac{3}{2} u u_x.
$$
A well known two-component case  is  the Kaup-Boussinesq system, see e.g. \cite[eqn. (4)]{pav}:
$$
\begin{aligned}
u^1_t & = u^2_x - \frac{3}{2} u^1 u^1_x\, , \\ 
u^2_t & = \frac{m}{2} u^1_{xxx} - u^2 u^1_x - \frac{1}{2} u^1 u^2_x\, .
\end{aligned}
$$
where $m$ is an arbitrary constant.

We will also consider evolutionary systems of PDEs {\it with differential constraints.}  They are given by 
\begin{equation}\label{eq:1v}
\begin{aligned}
u^i_{t} & = \xi^i(u, q), \quad i = 1, \dots, n, \\
0 & = p (u, q). 
\end{aligned}    
\end{equation}
Here,  unknown functions are  $u(x,t) = (u^1(x,t), \dots, u^n(x,t))^\top$ and $q(x,t)$. Each $\xi^i(u, q)$ is a differential polynomial in $u^1,...,u^n$ and $q$,  whereas the {\it differential constraint } $p(u,q)$ is a differential polynomial in $q$ and a function in $u$ (but not in derivatives of $u$). In most cases we consider just one constraint; the only exception is Example \ref{eq:3.5}.

In all the systems that we construct and analyse, the differential constraint $p$  can be explicitly resolved w.r.t. at least one of $u^i$'s.  Solving $p(u,q)=0$ with respect to $u^i$  and substituting  $u^i=\tilde p(u^1,\dots, u^{i-1}, u^{i+1},\dots, u^n, q)$  in  the first equation of \eqref{eq:1v} allows us  to think of \eqref{eq:1v}  as a system  of $n-1$ evolutionary PDEs and one non-evolutionary PDE  
 on  $n$ unknown functions $u^1,\dots, u^{i-1}, u^{i+1},\dots, u^n$ and $q$.  
As an example, consider the  Dullin-Gottwald-Holm equation \cite{gdh} (in this case, $n=1$): 
\begin{equation}\label{eq:ch}
\begin{aligned}
u_t & = \frac{\gamma}{2}q_{xxx} + u q_x + \frac{1}{2}q u_x, \\
u & = q + \frac{m}{2} q_{xx}.
\end{aligned}
\end{equation}

Substituting the expression for $u$ into the first equation gives 
$$
q_t + \frac{m}{2} q_{xxt} = \frac{\gamma}{2} q_{xxx} + \frac{3}{2} qq_x + \frac{m}{2} q_{xx} q_x + \frac{1}{2} q q_{xxx}, 
$$
which  is a single non-evolutionary PDE  on   $q(x, t)$.

We will also deal with PDE systems given by \eqref{eq:1} with the right hand side $\xi^i$ being a formal differential series (i.e., infinite sum of monomials in derivative variables $u^j_x, u^j_{xx}, u^j_{xxx},\dots$ with coefficients being functions of $u^1,...,u^n$). We refer to such systems  as  {\it formal}  evolutionary PDEs  (systems of type  \eqref{eq:1} with $\xi^i$ being a differential polynomial will be called {\it non-formal}).

Let us explain, using the Dullin-Gottwald-Holm equation  \eqref{eq:ch}   as  an example,  
  the  relation between systems of type \eqref{eq:1v} and {  formal}  evolutionary PDEs.  Observe that  the second equation of  \eqref{eq:ch}  can be formally solved with respect to $q$:
	$$q =   u - \frac{m}{2} u_{xx} + \frac{m^2}{4} u_{xxxx} - \dots$$
	(we discuss neither convergence of this differential series nor boundary or other conditions). 
	Substituting this formal expression into the first equation of \eqref{eq:ch} leads to the formal evolutionary PDE
$$
u_t =  \frac{3}{2} u u_x + \frac{\gamma}{2} u_{xxx}  - \frac{m}{2} \Big( \frac{\gamma}{2}u_{xxxxx} - u u_{xxx} - \frac{1}{2} u_{xx} u_x \Big) + \dots,
$$
depending on $m$ as a parameter.

Recall that a (formal) evolutionary PDE-system  $u^i_\tau = \eta^i(u)$ is a {\it (formal)  symmetry} of a  (formal) evolutionary PDE-system   $u^i_t = \xi^i(u)$ (or, equivalently, these two PDE-systems  {\it  commute}) if the following {\it commutator} vanishes 
\begin{equation}\label{bracket}
\sum \limits_{j = 0}^\infty \sum_{\alpha=1}^n \Big( \pd{\xi^i}{u^\alpha_{x^j}} D^j \big(\eta^\alpha) - \pd{\eta^i}{u^\alpha_{x^j}} D^j \big(\xi^\alpha) \Big) =0.
\end{equation}
Here and below $u_{x^j}=u_{\underset{j \ \mathrm{times}}{xx...x}}$,  $D =\tfrac{d}{dx} $ is the  {\it total derivative}  in $x$ (for example, $D(u^\alpha u^\alpha _x)= (u^\alpha_x)^2 + u^\alpha u^\alpha_{xx}$) 
 and $D^j$ stands for the $j$-th power of $D$ (for example  $D^3(u^\alpha)=u^\alpha_{xxx}=u^\alpha_{x^3}$). In the  case of usual (= non-formal) PDEs,  $u^i_\tau = \eta^i(u)$ is a symmetry of $u^i_t = \xi^i(u)$  if an only if 
(at least in the analytic category)   both systems can be solved simultaneously, that is, there exist $n$ functions $u^i(x, t, \tau)$, which satisfy all  $2n$ PDEs $u^i_\tau = \eta^i(u)$ and  $u^i_t = \xi^i(u)$ .

A (formal) differential series $\mathrm v$ is said to be a {\it conservation law density}  of  the (formal) evolutionary equation $u^i_t = \xi^i(u)$ if
\begin{equation}
\label{conserve}
\mathrm v_t = \sum \limits_{j = 0}^\infty \sum_{\alpha=1}^n \pd{\mathrm v}{u^\alpha_{x^j}} D^j\big(\xi^\alpha\big) = D\mathrm w     
\end{equation}
for some formal differential series $\mathrm w$.  Such  $\mathrm v$ is defined up to addition of an arbitrary total derivative $D\tilde{\mathrm v}$. In the non-formal case, $\mathrm v$ should be understood as the  {\it density} of a conservation law: in this case, for any fast decaying or periodic 
 solution $u(t,x)$,  the integral $\int \mathrm v(u, u_x,\dots) \ddd x$   is independent of $t$ for any solution $u(x,t)$ of \eqref{eq:1}. 
We notice that both formulas \eqref{bracket} and \eqref{conserve}  `respect' the degree of differential monomials so that they are well defined for formal differential series.

The main result  of our paper is a construction  of a family of multi-component integrable PDE systems of the form \eqref{eq:1} and \eqref{eq:1v}. There are several different notions of integrability in this context in the literature. In the case of systems   \eqref{eq:1}, we construct  infinite hierarchies of (non-formal)  conservation laws  and  of (non-formal)  pairwise commuting symmetries.   The differential degrees of the conservations laws and symmetries grow within each hierarchy.

By integrability of systems  \eqref{eq:1v}, we understand the existence of infinitely many independent 
(possibly, formal) commuting symmetries and conservations laws of the formal evolutionary equation obtained from
\eqref{eq:1v} by the method demonstrated in the example above; this can be done for all systems we construct.  
 
The paper  is organised as follows. In Section \ref{sect1.1} we describe our main examples. They are parameterised  by certain discrete and continuous parameters and are of four  types.
The equations of Type II and Type IV are evolutionary, whereas those of Types I  and III are evolutionary with differential constrains.   All  of them, however, can be considered as different representatives of one single integrable system understood as a commutative algebra of evolutionary flows, which includes both formal and non-formal flows.  We give  explicit recursive formulas for common (non-formal) symmetries and conservation laws for all these PDE systems,  see Theorem \ref{thm01}.  

All the famous integrable systems listed in the abstract correspond to  certain choices of parameters. In Section  \ref{sect2.3}, we specify those parameters for each of  them. 
We also obtain other known integrable systems, e.g.,  the  Marvan-Pavlov  system \cite{pm}, which we essentially generalise. The  generalisations are given by explicit formulas and have no low-component analogues.

Theorem \ref{thm01} follows directly from a  more general construction described in Section  \ref{sect2}, see  Theorems \ref{t1},  \ref{t2}  and  \ref{t3}. Namely, Theorem \ref{t1} constructs a family of  evolutionary PDEs with differential constrains, as well as 
  formal commuting symmetries and conservations laws for them. Next, Theorem \ref{t2}   explains how one can `cook' non-formal evolutionary PDEs, non-formal symmetries and non-formal conservation laws starting from those constructed in Theorem \ref{t1}.  The constructions in Theorems \ref{t1}, \ref{t2} depend on a solution of a certain, possibly infinite,  system of PDEs on a Nijenhuis manifold $(\mathrm M^n, L)$. In Theorem \ref{t3},  we solve  this system under the additional assumption that $L$ is {\it differentially non-degenerate} and, hence, come to the integrable systems from Theorem~\ref{thm01}. 
 
Let us now comment on the circle of ideas which led us to these results.  The construction was developed within the  {\it Nijenhuis Geomery} programme \cite{nij}.  Its main ingredient is a Nijenhuis operator $L$ on a manifold $\mathrm M^n$, that is a (1,1)-tensor field $L= L^i_j $ such that its Nijenhuis torsion vanishes, i.e., 
$$
 L^2 [\nu,\eta] - L[L\nu,\eta] - L[\nu,L\eta] + [L\nu,L\eta] = 0 
$$
for arbitrary vector fields $\nu, \eta$.  In  our   recent paper \cite{appnij3},  we constructed all non-degenerate pencils of compatible $\infty$-dimensional Poisson structures of type $\mathcal{P}_3 + \mathcal{P}_1$, where  the Poisson structure $\mathcal{P}_1$ has order  1 and $\mathcal{P}_3$ is a Darboux-Poisson structure of order 3. Magri-Lenard scheme applied to these pencils  leads to certain integrable bi-Hamiltonian systems. Translating them to the language of Nijenhuis Geometry allowed us to generalise our construction further and obtain integrable systems which are not necessarily  Hamiltonian. 
We view Nijenhuis Geometry as the most natural framework for them  and expect that  the systems  and their  properties can and should  be understood in the context of Nijenhuis operators, with no other geometric structure involved.

\section{Explicit formulas for new integrable systems, their symmetries and conservation laws}\label{sect1.1}

\subsection{Four types of equations}  

A Nijenhuis operator $L$ on $\mathrm M^n$ 
is called  differentially   non-degenerate, if the differentials of the coefficients of its characteristic polynomial are linearly independent  at every point. 
Typical examples of differentially   non-degenerate  Nijenhuis  operators are as follows: 
\begin{equation}\label{dnd}
L_{\mathsf{comp}} = \left( \begin{array}{ccccc}
     u^1 & 1 & 0 & \dots & 0  \\
     u^2 & 0 & 1 & \dots & 0  \\
     \vdots & & & \ddots & \\
     u^{n - 1} & 0 & 0 & \dots & 1 \\
     u^n & 0 & 0 & \dots & 0
\end{array}\right)  \  \ \textrm{ and }  \  \  L_{\mathsf{diag}} = \left( \begin{array}{ccccc}
     x^1 & 0 & 0 & \dots   & 0  \\
      0 & x^2  & 0 &  \dots  & 0  \\
       \vdots  &    & \ddots &  & \vdots \\
      0 & \dots  & \dots & x^{n-1} & 0  \\
      0& 0 & \dots  & \dots &  x^n 
\end{array}\right).    
\end{equation}
Here $u^1,...,u^n$  and $x^1,...,x^n$ are local coordinate charts on $\mathrm M^n$. Moreover, 
in the case of $L_{\mathsf{diag}}$  we assume that $x^i$ are all different. In fact, these two operators are locally  isomorphic: if we rewrite  $L_{\mathsf{diag}}$ in the coordinates $u^1(x),...,u^n(x)$ that are coefficients of its characteristic polynomial   
(more precisely, we set $\det( t \Id- L_{\mathsf{diag}})= t^n- u^1t^{n-1}-\dots -u^n$), then it transforms into
$L_{\mathsf{comp}}$. 
\weg{Every eigenvalue of $L_{\mathsf{comp}}$ has geometric multiplicity $1$ and   its minimal polynomial coincides with the characteristic one.} Note also that every differentially non-degenerate Nijenhuis operator $L$ reduces to the companion form $L_{\mathsf{comp}}$ by an appropriate coordinate transform, and to the diagonal form  $L_{\mathsf{diag}}$ near those points where $L$ has $n$ distinct real eigenvalues, see e.g. \cite{nij}. 
 
Choose $N\ge 0$ and consider natural numbers $n_0, n_1, \dots, n_N$ and $\ell_1,...,\ell_N$ with conditions $n_0 + n_1 + \dots + n_N=n=\dim \mathrm M$ and $n_0 - \ell_1 n_1 - \dots - \ell_N n_N = d \geq 0$, and fix a polynomial 
\begin{equation}
\label{eq:mpolynomial}
m(\lambda)=m_0 + m_1 \lambda + \dots + m_d \lambda^d
\end{equation} 
of degree $\le d$. These are the parameters of our construction: $2N + 1$ natural numbers and $d + 1$ coefficients $m_0, \dots, m_d$.

Next, take the direct product  $\mathrm M^n = U_0 \times \dots  \times U_N$  of $N+1$  discs  $U_0,...,U_N$  of dimensions $n_0, n_1,...,n_N$ equipped with differentially non-degerenerate Nijenhuis operators $L_0,...,L_N$. The operator field $L$ on $\mathrm M^n$ is defined as 
\begin{equation}
\label{eq:Ldirectsum}
L=L_0 \oplus \dots \oplus L_N.
\end{equation}
Consider  the  following  family of functions $\sigma(\lambda)$ on ${  M^n}$, depending on $\lambda$ as a parameter (in general $\lambda \in \mathbb C$, so the functions might be complex-valued): 
\begin{equation} 
\label{eq:7cbis1}
\sigma(\lambda) = \frac{\operatorname{det} (L_0 - \lambda \operatorname{Id})}{\operatorname{det} (L_1 - \lambda \operatorname{Id})^{\ell_1} \dots \operatorname{det} (L_N - \lambda \operatorname{Id})^{\ell_N}}
\end{equation}
Here  in each expression $\operatorname{det}(L_i - \lambda \operatorname{Id})^{\ell_i}$, the identity matrix $\Id$ is of the same size as $L_i$, i.e.,  $n_i \times n_i$ and $\ell_i$ denotes the power. 

Next, consider the vector field $\zeta_0$ on $\mathrm M^n$ uniquely defined by the relations
$$
\mathcal L_{\zeta_0} \operatorname{det}(L_0 - \lambda \operatorname{Id}) = 1, \quad \mathcal L_{\zeta_0} \operatorname{det}(L_i - \lambda \operatorname{Id}) = 0, \quad i = 1, \dots, N,
$$
where $\mathcal L_{\zeta_0}$ denotes the Lie derivative, and define another vector field $\zeta$ by setting 
\begin{equation}\label{pzeta}
    \zeta = p(L)\zeta_0.
\end{equation}
where $p(t) = \operatorname{det} (L_1 - \lambda \operatorname{Id})^{\ell_1} \dots \operatorname{det} (L_N - \lambda \operatorname{Id})^{\ell_N} m(t)$ is a polynomial in $t$  with coefficients being functions on $\mathrm M^n$. 

Based on these settings, we finally introduce four types of equations. As unknown functions, we consider $u^1(x,t),\dots, u^n(x,t), q(x,t)$, where $(u^1,\dots, u^n)$ should be understood as coordinates on $\mathrm M^n$ and $q$ as an additional function.

\textbf{Type I. } For each real or complex number $\lambda$ consider the  equations 
\begin{equation}\label{eq:t1}
\begin{aligned}
u_t & = q_{xxx} (L - \lambda \operatorname{Id})^{-1} \zeta + q (L - \lambda \operatorname{Id})^{-1} u_x, \\
0 & = m(\lambda) q_{xx} q - \frac{1}{2} m(\lambda) (q_x)^2 + \sigma(\lambda) q^2 - 1.
\end{aligned}
\end{equation}
This is a system of the form  \eqref{eq:1v}, i.e., an $n$-component system with a differential constraint.

\textbf{Type II. }If $\lambda_i$ is a root of $m(\lambda)$, then the differential constraint in Type I becomes   degenerate and takes the form  $\sigma(\lambda_i)q^2 = 1$. Thus, we get an evolutionary PDE
\begin{equation}\label{eq:t2}
\begin{aligned}
u_t & = \Bigg( \frac{1}{\sqrt{\sigma(\lambda_i)}}\Bigg)_{xxx} (L - \lambda_i \operatorname{Id})^{-1} \zeta + \Bigg( \frac{1}{\sqrt{\sigma(\lambda_i)}}\Bigg) (L - \lambda_i \operatorname{Id})^{-1} u_x. \\
\end{aligned}
\end{equation}

\textbf{Type III. } Consider the equation
\begin{equation}\label{eq:t3}
\begin{aligned}
u_t & = q_{xxx} \zeta + (L + q \operatorname{Id}) \, u_x, \\
\frac{1}{2} \Big( \operatorname{tr} L_0 & - \sum_{j = 1}^N \ell_j \operatorname{tr} L_j\Big) =  q + (-1)^d \, \frac{m_d}{2} \,q_{xx},
\end{aligned}
\end{equation}
where $m_d$ is the highest coefficients of the polynomial $m(t)= m_0 + m_1 t + \dots + m_d t^d$. This is, again, a system of of the form \eqref{eq:1v}, i.e., an $n$-component system with a differential constraint. 

\textbf{Type IV. } Assume now  $m_d = 0$   (in this case we say that  $m(t)$ has a root at infinity, the terminology will be clarified later).  Then \eqref{eq:t3}  takes the form
\begin{equation}\label{eq:t4}
\begin{aligned}
u_t = \frac{1}{2}\Big(\operatorname{tr} L_0 - \sum_{j = 1}^N \ell_j \operatorname{tr} L_j\Big)_{xxx} \zeta + \Big( L + \frac{1}{2}\Big(\operatorname{tr} L_0 - \sum_{j = 1}^N \ell_j \operatorname{tr} L_j\Big) \operatorname{Id} \Big)u_x.
\end{aligned}
\end{equation}


\subsection{Commuting flows and conservation laws for the equations of Types I -- IV}

Here we describe an explicit procedure that generates commuting symmetries and conservation laws for the above four types of equations. 

{\bf Step 1.} In the one-component case, consider the relation
\begin{equation}\label{ric1}
    \sigma = \frac{1}{2}\mathrm u^2 + \mathrm u_x.     
\end{equation}
and its formal solution $\mathrm u = \mathcal u_1 + \mathcal u_2 + \dots$ as a differential series in $\sigma$. 

The recursion formula for the components of $\mathrm u$ from \eqref{ric1} was essentially discovered by Kruskal and Miura. In the form we need (up to notation) it appeared e.g. in \cite[eqns. 2.16--2.19]{bhd}:  
\begin{equation}
\label{eq:rec}
    \mathcal u_1  =\sqrt{ 2 \sigma}, \quad
    \mathcal u_{i + 1}  = - \frac{1}{\mathcal u_1} \Bigg(\,\frac{1}{2} \sum \limits_{j = 2}^{i} \mathcal u_j \mathcal u_{i + 2 - j} + (\mathcal u_{i})_x\Bigg), \quad i \geq 1.
\end{equation}

In the expansion $\mathrm u = \sum \mathcal u_i$, we are interested in the odd terms only and introduce two formal differential series
\begin{equation}\label{first1}
{\mathrm v}(\sigma, m) = \sqrt{2} \sum_{s = 0}^{\infty} (- 1)^s m^s \mathcal u_{2s + 1}  
\end{equation}
and 
\begin{equation}\label{second1}
\mathrm w(\sigma, m) = \sqrt{2} \sum_{s = 0}^{\infty} (- 1)^{s} m^{s} \delta \mathcal u_{2s + 1}.    
\end{equation}
Here $\delta$ stands for the variational derivative w.r.t. $\sigma$ and $m$ is considered as a formal parameter.

\noindent {\bf Important fact.}
The series ${\mathrm w}(\sigma, m)$ satisfies the following identity
\begin{equation}\label{id}
    m {\mathrm w}_{xxx}(\sigma, m) + 2 \sigma {\mathrm w}_x(\sigma, m) + \sigma_x {\mathrm w}(\sigma, m) = 0.
\end{equation}
This identity is essentially known and comes from the theory of local infinite-dimensional Poisson structures.  It can be understood as the fact that ${\mathrm v}(\sigma, m)$ is a formal Casimir of the Poisson structure defined by the operator $mD^3 + 2\sigma D + \sigma_x$ (for details and proof see e.g. \cite[Proposition 2.1]{bhd} and also \cite[Theorem 2.4]{kon} for $n$-component systems).   One can also view \eqref{id} as a way of applying the Magri-Lenard scheme to produce commuting symmetries  for  the Harry Dym equation.  However, we do not need such an interpretation and will use identity \eqref{id} as it is.

We will also use the following crucial observation by Gelfand and Dikii  \cite{gdk} (see also \cite{fordy3}).  Multiplying the l.h.s. of \eqref{id} by $\mathrm w (\sigma,m)$ and integrating in $x$ we get another important identity for  $\mathrm w$ (cf. the differential constraint from \eqref{eq:t1}):
\begin{equation}
\label{eq:smth3}
m \,\mathrm w_{xx}  \mathrm w  - \frac{1}{2} \, m \,\mathrm w_{x}^2 + \sigma\, \mathrm w^2 = 1,
\end{equation}
which, in particular, allows us to reconstruct all the terms of  \eqref{second1} step by step starting from the first term $\sqrt{2} \, \delta \mathcal u_1 = \frac{1}{\sqrt{\sigma}}$.

{\bf Step 2.}  For  $m(\lambda)$ and $\sigma (\lambda)$  defined by \eqref{eq:mpolynomial} and \eqref{eq:7cbis1} respectively,
consider the formal differential series (in any local coordinates $u^1,\dots, u^n$ on $\mathrm M^n$) with components depending on parameter $\lambda$  
\begin{equation}
\label{collection1}
\mathrm v(\lambda) = \mathrm v(\sigma(\lambda;u), m(\lambda)), \quad \mathrm w(\lambda) = \mathrm w(\sigma(\lambda;u), m(\lambda)),     
\end{equation}
obtained by replacing  $\sigma$ with  $\sigma(\lambda)=\sigma(\lambda; u^1,\dots,u^n)$ and $m$ with $m(\lambda)$ in \eqref{first1} and \eqref{second1}. Next, introduce the formal $n$-component vector field 
\begin{equation}\label{fields1}
\xi(\lambda)  = \mathrm w(\lambda)_{xxx} \, (L - \lambda\Id)^{-1} \zeta  + \mathrm w(\lambda) (L - \lambda\Id)^{-1} u_x. 
\end{equation}

This vector field is naturally related to the r.h.s. of the first equation of system  \eqref{eq:t1} of Type 1.   In fact,  \eqref{eq:t1} can be equivalently rewritten as $u_t = \xi(\lambda)$.  It follows from the fact that the differential series  $\mathrm w(\lambda)$  satisfies the same differential relation \eqref{eq:smth3} as the function $q$ (see second equation of   \eqref{eq:t1}) and can be uniquely reconstructed from it (see `Important fact' in Step 1).

{\bf Step 3.}
For each root $\lambda_i$ of the polynomial $m(\lambda)$,   expand both $m(\lambda)$ and $\sigma(\lambda,u)$  into Taylor series in powers of $\varepsilon=\lambda - \lambda_i$, i.e.,
$$
m(\lambda_i + \varepsilon) =  \sum_{s=1}^\infty  m_{s,\lambda_i} \varepsilon^s, \quad \sigma(\lambda_i + \varepsilon, u)=
\sum_{s=0}^\infty \sigma_{s, \lambda_i} (u) \varepsilon^s 
$$

Substitution  $m(\lambda)=\sum_{s=1}^\infty  m_{s,\lambda_i} \varepsilon^s$ and 
$\sigma (\lambda; u)= \sum_{s=0}^\infty \sigma_{s, \lambda_i} (u) \varepsilon^s$   ``transforms'' $\mathrm v(\lambda)=\mathrm v(\sigma(\lambda;u), m(\lambda))$ and $\mathrm w(\lambda) = \mathrm v(\sigma(\lambda,u), m(\lambda))$, as well as $\xi(\lambda)$ defined by \eqref{fields1}, into series in powers of $\varepsilon$ :
\begin{equation}
\label{eq:vslambda1}
\mathrm v\bigl(\sigma(\lambda_i + \varepsilon,u),  m(\lambda_i + \varepsilon)\bigr) = \sum_{s=0}^\infty \mathcal v_{s,\lambda_i}\varepsilon^s, \quad   \mbox{with} \ \ \mathcal v_{0,\lambda_i} = 2 \sqrt{\sigma(\lambda_i)}.
\end{equation}
\begin{equation}
\label{eq:wslambda1}
\mathrm w(\sigma(\lambda_i + \varepsilon),  m(\lambda_i + \varepsilon)) = \sum_{s=0}^\infty \mathcal w_{s,\lambda_i}\varepsilon^s, \quad  \mbox{with} \ \  \mathcal w_{0,\lambda_i} = \frac{1}{\sqrt{\sigma(\lambda_i)}}.
\end{equation}
\begin{equation}
\label{eq:xislambda1}
\xi = \sum_{s=0}^\infty \xi_{s,\lambda_i} \varepsilon^s,    
 \mbox{with} \ \ \xi_{0,\lambda_i} = \left( \tfrac{1}{\sqrt{\sigma(\lambda_i)}}\right)_{xxx} (L - \lambda\Id)^{-1}  \zeta  + 
\tfrac{1}{\sqrt{\sigma(\lambda_i)}} \, (L - \lambda\Id)^{-1} u_x.
\end{equation}

Similarly for $\lambda=\infty$, in the above construction we replace  $\sigma(\lambda)$ and $m(\lambda)$ with 
$\bar\sigma(\lambda) = (-\lambda)^d \sigma\left( \frac{1}{\lambda}\right)$,
$\bar m(\lambda) = (-\lambda)^d m\left( \frac{1}{\lambda}\right)$ and also $L -\lambda\Id$ with $\Id -\lambda L$. 
In particular, we set
$$
\begin{aligned}
\bar {\mathrm v}(\lambda) &= \mathrm v (\bar\sigma(\lambda), \bar m(\lambda)), \quad \bar {\mathrm w}(\lambda) = \mathrm w (\bar\sigma(\lambda), \bar m(\lambda)), \\
\bar \xi(\lambda)  & = \bar{\mathrm w}(\lambda)_{xxx} \, (\Id - \lambda L)^{-1} \zeta  + \bar{\mathrm w}(\lambda) (\Id - \lambda L)^{-1} u_x. 
\end{aligned}
$$

 Then if $\bar m(0) = 0$, we substitute  $\bar\sigma(\varepsilon)$, $\bar m(\varepsilon)$ into these relations  and expand in powers of $\varepsilon$ to get\footnote{In these power series, we shift indices of  coefficients by 1.  The reason is that the first terms of these expansions are trivial and can be ignored.  This shift also allows us to keep notation consistent with the case of $\lambda_i\ne \infty$.} 
 
\begin{equation}
\label{eq:vsinfty}
\bar{\mathrm v}(\varepsilon)= \sum_{s=0}^\infty \mathcal v_{s - 1,\infty}\varepsilon^s, \quad   \mbox{with} \ \ \mathcal v_{- 1,\infty} = 2 \sqrt{\bar\sigma(0)} = 2  \mbox{ and }   \mathcal v_{0,\infty} = - 2 f(u);
\end{equation}
\begin{equation}
\label{eq:wsinfty}
\bar{\mathrm w}(\varepsilon) = \sum_{s=0}^\infty \mathcal w_{s - 1,\infty}\varepsilon^s, \quad  \mbox{with} \ \  \mathcal w_{-1,\infty} = 1 \mbox{ and } \mathcal w_{0,\infty} = \frac{1}{2} f(u);
\end{equation}
\begin{equation}
\label{eq:xislambda11}
\bar \xi(\varepsilon) = \sum_{s=0}^\infty \xi_{s-1,\infty} \varepsilon^s,    \quad \mbox{with $\xi_{-1,\infty} = u_x$ and $\xi_{0,\infty}= \frac{1}{2}f_{xxx} \zeta + (L +\frac{1}{2} f\Id) u_x$}
\end{equation}
where $f$ is defined from  $\bar\sigma(\varepsilon) = 1 - \varepsilon f + \dots$, i.e., $ f = \frac{d}{d\varepsilon}|_{\varepsilon=0}  \bar\sigma(\varepsilon)=  \operatorname{tr} L_0 - \sum_{j = 1}^N \ell_j \operatorname{tr} L_j$. 

The coefficients of the $\varepsilon$-expansions \eqref{eq:vslambda1}, \eqref{eq:xislambda1}, \eqref{eq:vsinfty} and \eqref{eq:xislambda11}  define  hierarchies of common commuting symmetries and conservation laws for the above introduced equations of types I\,--\,IV. Namely, we have 

\begin{Theorem}\label{thm01}
Let $\lambda_1,\dots, \lambda_r$ be the roots of the polynomial $m(\lambda)$ (including $\infty$ when appropriate). Then the evolutionary PDE systems
$$
u_t = \xi_{s, \lambda_i}, \quad i = 1, \dots, r, \ \ s = 0, 1, \dots.
$$
are commuting symmetries and the differential polynomials
$$
\mathcal v_{s,\lambda_i}, \quad i = 1, \dots, r, \ \ s = 0, 1, \dots
$$
are conservation law densities for equations \eqref{eq:t1}--\eqref{eq:t4} of Types  I\,--\,IV.  Moreover,  the equations {\rm(\ref{eq:t2}, Type II)} and {\rm (\ref{eq:t4}, Type IV)} take the form  $u_t=\xi_{0,\lambda_i}$ for $\lambda_i \neq \infty$ and $\lambda_i = \infty$ respectively.
\end{Theorem}

Thus, for a Nijenhuis operator $L$ (decomposed into differentially non-degenerate blocks),  Theorem 1 gives a series of multi-component integrable systems and provide, for each of them, commuting symmetries and conservation laws that can be constructed by an explicit iterative procedure.

\section{General construction}\label{sect2}


\subsection{Parameters of the general construction}

Let $L$ be a Nijenhuis operator on $\mathrm M^n$ and $f:\mathrm M^n \to \R$ be a function such that the 1-form $L^* \ddd f$ is closed so that locally  $L^* \ddd f = \ddd f_1$ for some function $f_1$. Then (see Section 3 in \cite{magri1}) there exists an infinite sequence of functions $f_k$, $k = 1, \dots$,  such that $\ddd f_k = (L^*)^k \ddd f$. We refer to $f$ as a {\it conservation law} for the Nijenhuis operator $L$. The name comes from the fact that  $f$ provides a conservation law in the sense of \eqref{conserve} for the quasilinear system $u_t = L u_x$.  The above observation means that  $f$ is a conservation laws for every power of $L$ or, equivalently, generates a hierarchy of conservation laws for $L$.

In particular, this implies that $(L^* - \lambda \operatorname{Id})^{-1} \ddd f$ is also closed for any $\lambda$. Thus, there locally exists a function $g(\lambda; u)$ depending on $\lambda$ as a parameter and such that $\ddd g = (L^* - \lambda \operatorname{Id})^{-1} \ddd f$. Here $u=(u^1,\dots, u^n)$ are coordinates on $\mathrm M^n$ and $\ddd$ denotes the differential of a function w.r.t. $u$. One of the main ingredients of our construction is the function  $\sigma (\lambda; u) = e^{g(\lambda; u)}$ satisfying the identity
\begin{equation}
\label{eq:05}
\big(L^* - \lambda \operatorname{Id}\big) \ddd \sigma = \sigma \, \ddd f, \qquad \lambda \in \mathbb C.
\end{equation}

If  $L(u) -\lambda\Id$ is invertible, i.e. $\lambda\notin \operatorname{Spectrum} L(u)$, then $\sigma (\lambda; u)$ is analytic in $\lambda$,  otherwise the point $(\lambda, u)$ may be singular  (like pole or zero, or branching point).  

Next, assume that there exist  a vector field $\zeta$ on $M^n$  and constant  $C \in \R$ such that\footnote{This condition is quite non-trivial so that the existence of a non-zero $\zeta$ depends on $L$ and $f$.  However, for $\zeta = 0$ the construction still makes sense but reduces to a hydrodynamic type system, see Corollary \ref{cor1}.}
\begin{equation}\label{zeta2}
\mathcal L_\zeta \bigl(\sigma(\lambda; u)\bigr) + C \sigma(\lambda; u) = m(\lambda),
\end{equation}  
that is, the left hand side does not depend on $u$ and hence is a certain function of $\lambda$  (in the examples discussed below, $m(\lambda)$ is always a polynomial of degree $\le n=\dim M$). This triple, i.e., Nijenhuis operator $L$, conservation law $f$ and vector field $\zeta$ are parameters of the construction. 

Notice that $\sigma(\lambda;u)$ and $m(\lambda)$ satisfying  \eqref{eq:05}, \eqref{zeta2} are defined up to simultaneous multiplication by an arbitrary function $c(\lambda)$.  This kind of scaling is not important for the construction below and we will treat it as a trivial transformation.


\subsection{Main theorems}\label{sect:3.2}

Fix a triple $L$, $f$ and $\zeta$. Construct $\sigma(\lambda; u)$ and $m(\lambda)$ by formulas \eqref{eq:05} and \eqref{zeta2}. Using them construct infinite differential series $\mathrm v(\lambda)$ and $\mathrm w(\lambda)$ by  \eqref{first1} and \eqref{second1}. 

Recall that the series $\mathrm w(\lambda)$ satisfies the Gelfand-Dikii identity \eqref{eq:smth3}:
\begin{equation}\label{r3}
m(\lambda) \left(\mathrm w_{xx} (\lambda) \mathrm w (\lambda) - \frac{1}{2} (\mathrm w_{x} (\lambda))^2\right) + \sigma(\lambda) (\mathrm w(\lambda))^2 = 1.
\end{equation}
Based on this information,  we introduce an $n$-component system with a differential constraint\footnote{The equation of Type I from Introduction is exactly of this kind with $q=\mathrm{w}(\lambda)$ for a specific choice of parameters $L$, $f$ and $\zeta$.}
\begin{equation}\label{eq:main}
\begin{aligned}
    u_{t_\lambda} & = \mathrm w_{xxx}(\lambda) \big(L - \lambda \operatorname{Id}\big)^{-1} \zeta + \mathrm w (\lambda) \big(L - \lambda \operatorname{Id}\big)^{-1} u_x, \\
    0 & = m(\lambda)\left( \mathrm w_{xx} (\lambda) \mathrm w (\lambda) - \frac{1}{2} (\mathrm w_{x} (\lambda))^2 \right)+ \sigma(\lambda) (\mathrm w(\lambda))^2 - 1.
\end{aligned}
\end{equation}

\weg{Formula \eqref{eq:main} defines a continuous family of equations of type \eqref{eq:1v} (compare with Type I equation \eqref{eq:t1}).} 

In this construction one can naturally make sense of $\lambda = \infty$.  Namely, we define
$\bar\sigma (\lambda) = c(\lambda) \sigma (\frac{1}{\lambda})$ and $\bar m (\lambda) = c(\lambda) m (\frac{1}{\lambda})$, where $c(\lambda)$ is a suitable scaling factor. The function $\bar \sigma(\lambda)$ does not satisfy \eqref{eq:05}, but does satisfy a very similar relation
\begin{equation}\label{eq:05i}
    (\Id -\lambda L) \ddd \bar\sigma(\lambda) = - \lambda \bar \sigma (\lambda) \ddd f.
\end{equation}
This implies that  $c(\lambda)$ can be chosen in such a way that $\bar\sigma(\lambda)$ is analytic in $\lambda$ in a neighbourhood of zero and moreover, $\bar\sigma (\lambda) =  1 - \lambda f +  \dots$,   where dots denote higher order terms in $\lambda$. We will assume that $c(\lambda)$ is chosen in this way. Then we set $\bar{\mathrm v}(\lambda) = \mathrm v(\bar \sigma(\lambda; u), \bar m(\lambda))$ and $\bar{\mathrm w}(\lambda) = \mathrm w(\bar\sigma(\lambda; u), \bar m(\lambda))$ and rewrite the family of equations \eqref{eq:main} in the following equivalent form obtained by replacing $\lambda$ with  $\frac{1}{\lambda}$:
\begin{equation}\label{eq:formal2}
\begin{aligned}
    u_{\bar t_\lambda} & = \bar{\mathrm w}_{xxx}(\lambda) \big(\operatorname{Id} - \lambda L\big)^{-1} \zeta + \bar{\mathrm w}(\lambda) \big(\Id - \lambda L \big)^{-1} u_x, \\
    0 & = \bar m(\lambda) \left(\bar{\mathrm w}_{xx} (\lambda) \bar{\mathrm w} (\lambda) - \frac{1}{2} (\bar{\mathrm w}_{x} (\lambda))^2 \right) + \bar \sigma(\lambda) (\bar{\mathrm w}(\lambda))^2 - 1.
\end{aligned}
\end{equation}
More precisely, we have $\bar {\mathrm w}(\lambda) = \frac{1}{\sqrt{c(\lambda)}} \mathrm w(\frac{1}{\lambda})$ so that $u_{\bar t_\lambda}$ coincides with $u_{t_{\lambda^{-1}}}$ up to appropriate rescaling  (with a factor depending on $\lambda$).

This transformation allows us to set $\lambda = 0$ in  \eqref{eq:formal2}  which will naturally correspond to $\lambda = \infty$ in \eqref{eq:main}.  In particular, we set
\begin{equation}
\label{eq:vinfty1}
\mathrm v(\infty) = \bar {\mathrm v}(0)\quad\mbox{and}\quad    \mathrm w(\infty) = \bar {\mathrm w}(0).  
\end{equation}

However,  $\bar\sigma(0)=1$ leading to $\bar {\mathrm w}(0) = 1$ and hence to the trivial evolutionary equation $u_{\bar t_\lambda} = u_x$.  To get a non-trivial equation `at infinity',  we may consider the derivative of \eqref{eq:formal2}  at $\lambda = 0$, namely we set
$$
\begin{aligned}
u_{t_\infty} &=  \lim_{\lambda\to 0}  \frac{1}{\lambda} (u_{\bar t_\lambda} - u_{\bar t_0})  = \\
&=\frac{1}{\lambda}\left(\bar{\mathrm w}_{xxx}(\lambda) \big(\operatorname{Id} - \lambda L\big)^{-1} \zeta + \bar{\mathrm w}(\lambda) \big(\Id - \lambda L \big)^{-1} u_x - u_x\right) = \\
&= q_{xxx} \zeta  + (L + q\Id) u_x,
\end{aligned}
$$
where $\bar{\mathrm w}(\lambda) = 1 + \lambda q + \mbox{\small{(higher order terms in $\lambda$)}}$.  In other words, $q$ is the derivative of $\bar{\mathrm w}(\lambda)$ w.r.t. $\lambda$ at zero.  Substituting $\bar{\mathrm w}(\lambda) = 1 + \lambda q + \dots$ into the second equation of \eqref{eq:formal2} we obtain the following constraint for $q$:
$$
\bar m(0) q_{xx} + 2q - f(u) = 0.
$$
To summarise,  for $\lambda = \infty$ we consider the following evolutionary PDE system with a constraint:
\begin{equation}
\label{eq:formal3}
\begin{aligned}
u_{t_\infty} = q_{xxx} \zeta  + (L + q\Id) u_x, \\
0=\bar m(0) q_{xx} + 2q - f(u).
\end{aligned}
\end{equation}

Equations \eqref{eq:main} and \eqref{eq:formal3} (related to $\lambda \neq \infty$ and $\lambda=\infty$ respectively) are now understood as a parametric family with $\lambda \in \bar{\mathbb C}=\mathbb C \cup \{\infty\}$.  The main property of this family of PDEs is as follows.

\begin{Theorem}\label{t1}
Let $L$ be a Nijenhuis operator and $f$ a conservation law of $L$.  Consider $\sigma(\lambda; u)$ constructed from \eqref{eq:05}, and a vector field  $\zeta$ satisfying \eqref{zeta2} for a certain function $m(\lambda)$. Then for any $\lambda,\mu \in \bar{\mathbb C}=\mathbb C \cup \{\infty\}$,     the differential series $\mathrm{v}(\mu)$ defined by \eqref{first1}, \eqref{eq:vinfty1} is a conservation law density for the evolutionary flow $u_{t_\lambda}$ with a differential constraint defined by \eqref{eq:main}, \eqref{eq:formal3}.
Moreover, if $f$ is generic in the sense that $\ddd f, L^*\ddd f, \dots , (L^{n-1})^*\ddd f$ are linearly independent,  then the flows 
$u_{t_\lambda}$'s pairwise commute. 
\end{Theorem}

As a straightforward corollary of this construction, we may consider the `trivial' case when $\zeta = 0$ and $m(\lambda)=0$.  In this situation,  the first term in \eqref{eq:main} disappears, but  our construction still gives a non-trivial series of integrable quasilinear systems.

\begin{Corollary}\label{cor1}
Let $L(u)$ be a Nijenhuis operator, $f(u)$ a conservation law of $L$ and $\sigma(\lambda; u)$ denote the function satisfying 
\eqref{eq:05}.  Then the evolutionary flows 
$$
u_{t_\lambda} = \sigma (\lambda; u) \bigl(L(u)- \lambda \Id\bigr)^{-1} u_x
$$
pairwise commute for all $\lambda$'s.  Moreover, the functions $\dfrac{1}{\sigma(\mu,x)}$ are common conservation law densities for these flows (for all $\lambda$ and $\mu$). 
\end{Corollary} 

\begin{Remark} {\rm  
In the assumptions of Theorem \ref{t1},  we obtain the flows of the form  $\mathrm w(\lambda)(L-\lambda\Id)^{-1} u_x$ with $\mathrm w (\lambda)= \sigma(\lambda)^{-1/2}$ for $m=0$. This exponent $-\frac{1}{2}$, however, is not very essential. 
Indeed, if $\sigma(\lambda)$ satisfies \eqref{eq:05}, then $\sigma(\lambda)^c$ satisfies \eqref{eq:05} also with $f$ replaced with $\tilde f = c\cdot f$, so that Corollary \ref{cor1}  can be easily obtained by an appropriate rescaling.  Of course, this corollary admits a direct proof without using Theorem \ref{t1}.  
}\end{Remark} 

We also note that Corollary \ref{cor1} can be understood as a $\lambda$-version of the construction by 
F. Magri suggested in \cite{magri1} and then developed in \cite{lm}. If $f=\tr L$ and $\lambda \to \infty$, then we obtain the system studied by E. Ferapontov and M. Pavlov in \cite{pf}  (see also \cite{appnij2})  and for $f = c\cdot \tr L$, $c\in\R$   we obtain the so-called $\varepsilon$-systems studied by M. Pavlov \cite{Pavlov}.

We will need another corollary from Theorem \ref{t1}. Consider a formal PDE  
\begin{equation}
\label{eq:xiformal}
u_{t_\lambda}=\xi(\lambda)
\end{equation}
obtained from \eqref{eq:main}  by resolving the constraint w.r.t. $\mathrm w(\lambda)$, i.e., expressing $\mathrm w(\lambda)$  as a formal differential series and substituting it into the first equation of \eqref{eq:main}.  As a result,  the r.h.s. of \eqref{eq:xiformal}  becomes  a formal differential series in the derivatives  $u_x, u_{xx}, \dots$ whose coefficients are functions in $\lambda$ and $u$.  

We now fix $\lambda$ and, in a small neighbourhood of it,  expand $\xi (\lambda + \varepsilon)$ in powers of $\varepsilon$: 
\begin{equation}
\label{eq:expandxi}
\xi (\lambda + \varepsilon) = \sum_{s=0}^\infty \xi_{s,\lambda} \varepsilon^s.
\end{equation}
In the same way we defines $\varepsilon$-expansions for $\mathrm v(\lambda)$ and $\mathrm w(\lambda)$:
\begin{equation}
\label{eq:expandvw}
\mathrm v(\lambda + \varepsilon) = \sum_{s=0}^\infty \mathrm v_{s, \lambda} \varepsilon^s, \quad 
\mathrm w(\lambda + \varepsilon) = \sum_{s=0}^\infty \mathrm w_{s, \lambda} \varepsilon^s.
\end{equation}
Notice that by construction, each coefficient $\xi_{s,\lambda}$, $\mathrm v_{s,\lambda}$ or $\mathrm w_{s,\lambda}$ is still a formal differential series in $u_x, u_{xx}, \dots$.

\begin{Corollary}\label{cor0}
In the settings of Theorem \ref{t1}, assume that the conservation law $f$ is generic in the sense that $\ddd f, L^*\ddd f, \dots , (L^{n-1})^*\ddd f$ are linearly independent.  Then the (formal) evolutionary flows defined by the (formal) vector fields 
$\xi_{s, \lambda}$ ($\lambda\in \bar{\mathbb C}$, $s=0,1,2,\dots$) pairwise commute. Moreover, $\mathcal v_{r, \mu}$
are common (formal) conservation law densities for all of them ($\mu\in \bar{\mathbb C}$, $r=0,1,2,\dots$).
\end{Corollary}

The next theorem is closely related to Corollary \ref{cor0} and deals with degeneration of the differential constraints that we observed in Section \ref{sect1.1} for Type I and Type III, but now in the general case. 

\begin{Theorem}\label{t2}
In the settings of Theorem \ref{t1},  let  $\lambda_i$  be a zero of $m(\lambda)$, i.e., $m(\lambda_i) = 0$  $(i\in\{1,2,\dots, k\})$.
Then all the coefficients $\xi_{s,\lambda_i}$, $\mathrm v_{s,\lambda_i}$ and $\mathrm w_{s,\lambda_i}$ of $\varepsilon$-expansions  \eqref{eq:expandxi} and \eqref{eq:expandvw} at the point $\lambda_i$ are differential polynomials  so that $u_t = \xi_{s,\lambda_i}$ is a usual evolutionary equation as in \eqref{eq:1}.  In particular, for $s=0$ these equations have the following form
\begin{equation}
\label{eq:explform}
u_{t_{\lambda_i}} = \xi (\lambda_i)  =  \left(\frac{1}{\sqrt {\sigma (\lambda_i)}}\right)_{xxx}  (L - \lambda_i\Id)^{-1}\zeta + 
\frac{1}{\sqrt {\sigma (\lambda_i)}} (L - \lambda_i\Id)^{-1} u_x,   \qquad \mbox{for $\lambda_i \ne \infty$}
\end{equation}
and
\begin{equation}
\label{eq:explforminfty}
u_{t_{\infty}} = \xi(\infty) =  \tfrac{1}{2} f_{xxx} \, \zeta + \left(L + \tfrac{1}{2} f \Id\right) u_x,      \qquad \mbox{for $\lambda_i = \infty$}.
\end{equation}
\end{Theorem}

Summarising the statements of Theorems \ref{t1}, \ref{t2} and Corollary \ref{cor0} we come to the following conclusion.  
For each $\lambda \in \bar {\mathbb C} = \mathbb C \cup \{\infty\}$ we define an evolutionary multi-component PDE system \eqref{eq:main}, \eqref{eq:formal3} with a differential constraint as in \eqref{eq:1v}.  The corresponding (formal) evolutionary flows $u_{t_{\lambda}} = \xi (\lambda)$ pairwise commute and admit an infinite family of common (formal) conservation laws also parameterised by $\lambda \in \bar {\mathbb C}$.  For some special values of the parameter $\lambda$, namely for the zeros $\lambda_1, \lambda_2,\dots$ of the function $m(\lambda)$ ($\infty$ is also allowed when appropriate),  the corresponding PDE equations $u_{t_{\lambda_i}} = \xi (\lambda_i)$ are {\it usual} evolutionary multi-component PDEs whose r.h.s. are differential polynomials as in \eqref{eq:1}.  Each $\lambda_i$ generates hierarchies of commuting {\it non-formal} symmetry fields  $\xi_{s,\lambda_i}$ and {\it non-formal} conservation laws $\mathrm v_{s,\lambda_i}$ for the whole family $u_{t_{\lambda}} = \xi (\lambda)$ of {\it formal} PDE systems.  Moreover, the members of these hierarchies are defined by means of an explicit iterative procedure in terms of the function $\sigma(\lambda; u)$ and vector field $\zeta$.

Thus, Theorems \ref{t1} and \ref{t2} give a recipe for constructing multi-component integrable PDEs starting from a Nijenhuis operator  $L$ and its conservation law $f$ satisfying certain conditions. However, in order to construct a specific example of such a system, we need to find  a function $\sigma(u, \lambda)$ and a vector field $\zeta(u)$ satisfying \eqref{eq:05} and \eqref{zeta2}, that is, to solve a (possibly, infinite) system of PDEs.  It is straightforward to check that  the function  $\sigma$ and vector field $\zeta$ given  by \eqref{eq:7cbis1}  and  \eqref{pzeta}  in  Section \ref{sect1.1}  are solutions of \eqref{eq:05} and \eqref{zeta2}. The construction from  Theorems \ref{t1}, \ref{t2}  applied to these $\sigma$ and $\zeta$  gives the   integrable systems of  Types I--IV from Section \ref{sect1.1} so that Theorem \ref{thm01} immediately follows. 


The next theorem shows that in the differentially non-degenerate case,   $\sigma$ and $\zeta$ given  by \eqref{eq:7cbis1}  and  \eqref{pzeta}  provide the only non-trivial solution of \eqref{eq:05} and \eqref{zeta2}.

\begin{Theorem}\label{t3}  Let $L$ be a differentially non-degenerate Nijenhuis operator and $f$ a conservation law of $L$ such that at a point $p\in  M^n$ the 1-forms $\ddd f,L^*\ddd f,\dots ,(L^{n-1})^{*} \ddd f$ are linearly independent.   Assume that there exist $\sigma(\lambda; u)$, $m(\lambda)$ and $\zeta$ satisfying \eqref{eq:05} and \eqref{zeta2} with $m(\lambda)\ne 0$. Then, in a small neighborhood of $p$, the Nijenhuis operator $L$, functions\footnote{We recall that $\sigma(\lambda; u)$ and $m(\lambda)$ are defined up to simultaneous multiplication by an arbitrary function $c(\lambda)$ and this freedom is assumed here.} $\sigma(\lambda; u)$, $m(\lambda)$  and  vector field $\zeta$  are as in Section \ref{sect1.1}, see \eqref{eq:Ldirectsum}, \eqref{eq:7cbis1}, \eqref{eq:mpolynomial} and \eqref{pzeta} respectively. 
\end{Theorem}

As already mentioned above, 
Theorem \ref{thm01}   follows directly from  Theorems  \ref{t1}, \ref{t2} by taking    $\sigma(\lambda)$   and $\zeta$ given by \eqref{eq:7cbis1}  and  \eqref{pzeta}.   Theorems \ref{t1} and \ref{t2}  are proved  in Section \ref{sect3} and Theorem \ref{t3}  in Section \ref{sect5}.


\subsection{Parameters  corresponding to known integrable systems}\label{sect2.3}

In this section we show that for particular choice of the parameters,  Type I -- IV equations from Section \ref{sect1.1} contain many famous integrable systems so that our approach allows one to generate a vast amount of different integrable systems in a unifying manner.

\begin{Ex}[KdV, Camassa-Holm, Dullin-Gottwald-Holm and their generalisations]\label{kdv}
\rm{
In dimension $n = 1$, the differentially non-degenerate Nijenhuis operator is   $L = u$. Due to Theorem \ref{t3}, the only possible $\sigma(\lambda)$ is $u - \lambda$ and then $\zeta = m_0 + m_1 u$.  Notice that 
$$
\mathcal L_\zeta \sigma(\lambda) = m_0 + m_1 u = m_0 + \lambda m_1 + m_1 (u - \lambda) = m(\lambda) + m_1 \sigma(\lambda),
$$
as required by \eqref{zeta2}.

The Type I equation in this case is
$$
\begin{aligned}
u_t & = q_{xxx} \frac{m_0 + m_1 u}{u - \lambda} + q \frac{u_x}{u - \lambda}, \\
0 & = (m_0 + \lambda m_1) q_{xx} q - \frac{1}{2} (m_0 + \lambda m_1) (q_x)^2 + (u - \lambda) q^2 - 1.
\end{aligned}
$$
This is a three-parameter ($m_0$, $m_1$ and $\lambda$) family of integrable evolutionary PDEs with differential constraint. 

If $m_1 \neq 0$, then taking $\lambda_0 = - \frac{m_0}{m_1}$ we get Type II equation
$$
u_t = m_1 \Big(\frac{1}{\sqrt{u - \lambda_0}}\Big)_{xxx}  + \frac{u_x}{(u - \lambda_0)^{3/2}}.
$$

This is a two-parameter ($\lambda_0$ and $m_1 \neq 0$)\footnote{The parameter $\lambda_0$ is not essential unless we consider the limit as $\lambda_0 \to\infty$.} family of equations. For  $\lambda_0 = 0$ it yields (after rescaling) the reduction of the coupled Harry Dym equation \cite[eqn. 26a]{fordy2}. It also appeared in  \cite{ch} as the first flow of the inverse Camassa-Holm hierarchy (flow $m^{(0)}_t$ in Section ``Bihamiltonian structure'' from \cite{ch}).

The Type III equation takes the form  
$$
\begin{aligned}
u_t & = q_{xxx} (m_0 + m_1 u) + (u + q) u_x, \\
\frac{u}{2} & = - \frac{m_1}{2} q_{xx} + q.
\end{aligned}
$$
This is a two-parameter family of the PDEs with a constraint. Differentiating the constraint we get an expression $m_1 q_{xxx} = 2 q_x - u_x$. Substituting it into the first equation and renaming the coefficient we obtain Dullin-Gottwald-Holm equation \cite{gdh}.  The case  $m_0 = 0$ gives the Camassa-Holm equation.

Finally, the Type IV equation corresponds to $m_1 = 0$, leading to the celebrated KdV equation
$$
u_t = \frac{m_0}{2}u_{xxx} + \frac{3}{2} u u_x. 
$$
}
\end{Ex}

\begin{Ex}[Coupled KdV and Harry Dym,  Kaup-Boussinesq and Ito systems]\label{ckdv}
\rm{
Take an arbitrary  $n$ and  consider the (differentially non-degenerate)  Nijenhuis operator $L=L_{\mathsf{comp}}$  given by the first formula of \eqref{dnd}. In the notation of  Section \ref{sect1.1},  we take  $N = 0$, $n = \ell_0 = d$ and $m(\lambda) = m_n \lambda^n + m_{n - 1} \lambda^{n - 1} + \dots + m_0$. Then
$$
\begin{aligned}
\sigma(\lambda) & = \operatorname{det} (L - \lambda \operatorname{Id}) = (-1)^n (\lambda^n - u^1 \lambda^{n - 1} - \dots - u^n), \\
\zeta & =  (-1)^{n+1} \left( (m_{n - 1} + m_n u^1) \pd{}{u^1} + \dots + (m_0 + m_n u^n)\pd{}{u^n}\right).
\end{aligned}
$$
It is easy to check that $\mathcal L_\zeta \sigma(\lambda) = m(\lambda) - (-1)^n m_n \sigma(\lambda)$. For every root $\lambda_i$ of $m(\lambda)$,  the Type II equation is
$$
u_t = \Bigg( \frac{1}{\sqrt{\det (L - \lambda_i\Id)}}\Bigg)_{xxx}  (L - \lambda_i \operatorname{Id})^{-1} \zeta 
+  \frac{1}{{\sqrt{\det (L - \lambda_i\Id)}}\big)}(L - \lambda_i \operatorname{Id})^{-1}u_x.
$$
This is  a $(n + 1)$-parameter family of integrable equations (with $m_i$'s as parameters involved in the formula for $\zeta$ above). For the rather special case $m_n = 1$ and $m_0 = \dots = m_{n - 1} = 0$, we get $\lambda_i = 0$ and taking $L=L_{\mathsf{comp}}$ as in \eqref{dnd} we come to coupled Harry Dym equations 
$$
u_t =   \left(\frac{1}{\sqrt {u^n}}\right)_{xxx}  e_1 \, + \, \left(\frac{1}{\sqrt {u^n}}\right) L_{\mathsf{comp}}^{-1} \, u_x, \quad e_1=\begin{pmatrix}  1 \\ 0 \\ \vdots \\ 0  \end{pmatrix}, \  u=\begin{pmatrix}  u^1 \\ u^2 \\ \vdots \\ u^n  \end{pmatrix}, \  
$$ 
introduced in \cite{fordy2} by  M. Antonowicz and A. Fordy.    

If $m_n = 0$, then $m(\lambda)$ has a root at infinity and Type IV equation is
$$
u_t = \tfrac{1}{2}\big(\operatorname{tr} L\big)_{xxx} \zeta + \big(L +\tfrac{1}{2}\operatorname{tr} L \operatorname{Id}
 \big) u_x.
$$
This is a family of integrable multi-component PDE systems with $n$ parameters $m_0, \dots, m_{n-1}$. More specifically, for $L=L_{\mathsf{comp}}$ given by \eqref{dnd} we get
$$
u_t = \tfrac{1}{2} \, u^1_{xxx} \, \zeta + \big(L_{\mathrm{comp}} +\tfrac{1}{2} \, u^1 \operatorname{Id} 
\big) u_x, \quad \zeta =  \begin{pmatrix}  m_{n-1} \\ \vdots \\ m_1 \\ m_0   \end{pmatrix} =\sum_{i=1}^n m_{n-i} e_i, \ m_i\in\mathbb R.    
$$ 
For  $\zeta =   e_i$, $i=1,\dots, n$, we get $n$ different systems known as coupled KdV systems and introduced by Antonowicz and Fordy in \cite{fordy}. 

The latter have two important examples for $n = 2$. For $m_2 = m_1 = 0, m_0 \neq 0$ after coordinate change $u^1 \to - u^1, u^2 \to - u^2$ we get the Kaup-Boussinesq system \cite[eqn. (4)]{pav}:
$$
\begin{aligned}
u^1_t & = u^2_x - \frac{3}{2} u^1 u^1_x, \\
u^2_t & = \frac{m_0}{2} u^1_{xxx} - u^2 u^1_x - \frac{1}{2} u^1 u^2_x.
\end{aligned}
$$
For $m_0 = m_2 = 0, m_1 \neq 0$ after coordinate change $u^1 \to - u^1, u^2 \to - u^2$ the same formula yields Ito system \cite[eqn. (25)]{pav}:
$$
\begin{aligned}
u^1_t & = \frac{m_1}{2} u^1_{xxx} - \frac{3}{2} u^1 u^1_x + u^2_x, \\
u^2_t & = - u^2 u^1_x - \frac{1}{2} u^1 u^2_x.
\end{aligned}
$$
}
\end{Ex}

\begin{Ex}[Marvan-Pavlov system]
\rm{
Now consider a pair of differentially non-degenerate Nijenhuis operators $L_0, L_1$ in dimensions $n_0$ and $n_1$. Assume   $n_0 - n_1 = d \geq 0$ and  consider  coordinates $u^1, \dots, u^{n_0}$ and $v^1, \dots, v^{n_1}$ in which $L_0$ and $L_1$ are given by the first formula of \eqref{dnd}. In the notations of Section \ref{sect1.1},  take $N = 1$ and $\ell_1 = 1$. We get $m(\lambda) = m_d \lambda^d + m_{d - 1} \lambda^{d - 1} + \dots + m_0$. In these coordinates
$$
\sigma(\lambda) = (-1)^d\frac{\lambda^{n_0} - u^1 \lambda^{n_0 - 1} - \dots - u^{n_0}}{\lambda^{n_1} - v^1 \lambda^{n_1 - 1} - \dots - v^{n_1}}
$$
and
$$
\begin{aligned}
\zeta & = (-1)^d m_d \Bigg( \sum_{r = 1}^{n_1} v^r \pd{}{u^r} - \sum_{j = 1}^{n_0} u^j \pd{}{u^j}\Bigg) + \sum_{s = 1}^d (-1)^d m_{d - s} \Bigg(- \pd{}{u^s} + \sum_{j = 1}^{n_1} v^{j} \pd{}{u^{j + s}}\Bigg)
\end{aligned}.
$$
By direct computation we have $\zeta(\sigma(\lambda)) = m(\lambda) - (-1)^d \sigma(\lambda)$. For every root $\lambda_i$ of $m(\lambda)$, the Type II equation is
$$
\begin{aligned}
u_t = \Bigg( \sqrt{(-1)^d \frac{\lambda_i^{n_1} - \sum_{r = 1}^{n_1} v^r \lambda_i^{n_1 - r}}{\lambda_i^{n_0} - \sum_{j = 1}^{n_0} u^j \lambda_i^{n_0 - j}}}\Bigg)_{xxx}&  (L - \lambda_i \operatorname{Id})^{-1} \zeta + \\
& + \sqrt{(-1)^d \frac{\lambda_i^{n_1} - \sum_{r = 1}^{n_1} v^r \lambda_i^{n_1 - r}}{\lambda_i^{n_0} - \sum_{j = 1}^{n_0} u^j \lambda_i^{n_0 - j}}} (L - \lambda_i \operatorname{Id})^{-1} u_x.
\end{aligned}
$$
This is a $d$-parameter family of integrable equations (with the coefficients of $m(\lambda)$ as parameters involved into the formula for $\zeta$). If $m_d = 0$ and the infinity is a root of $m(\lambda)$, we get the Type IV equation
$$
u_t = \frac{1}{2} \big(\operatorname{tr}L_1 - \operatorname{tr}L_2\big)_{xxx}\zeta  + \Big(L + \frac{1}{2} 
\big(\operatorname{tr}L_1 - \operatorname{tr}L_2\big) \operatorname{Id}\Big) u_x.
$$
This is a $d$-parameter family of integrable equations. Taking $m_0 \neq 0$ and all $m_i = 0, i \geq 1$ yields example by M. Marvan and  M. Pavlov  \cite{pm, pav}. 
}\end{Ex}

\begin{Ex}[Two-component Camassa-Holm and Dullin-Gottwald-Holm systems]
\rm{
Fix $n = 2$ and consider the Nijenhuis operator of the from
$$
L = \left( \begin{array}{cc}
     2u^1 & u^2  \\
     u^2 & 0  
\end{array}\right).
$$
This operator is related to left-symmetric algebras and plays an important role in the linearization problem (see \cite{nij2} for details). We take
$$
\sigma(\lambda) = \operatorname{det} (L - \lambda \operatorname{Id}) = \lambda^2 - 2 u^1 \lambda - (u^2)^2, \quad m(\lambda) = m_2 \lambda^2 + m_1 \lambda + m_0.
$$
The vector field $\zeta$ is
$$
\zeta = - \Big(\frac{m_1}{2} + m_2 u^1\Big) \pd{}{u^1} - \Big(\frac{m_0}{2 u^2} + \frac{m_2}{2} u^2\Big) \pd{}{u^2}.
$$
We get $\mathcal L_\zeta(\sigma(\lambda)) = m(\lambda) - m_2 \sigma(\lambda)$. For $m_2 \neq 0$, the Type III system in this setting is
$$
\begin{aligned}
u^1_t & = q_{xxx} \Big( - \frac{m_1}{2} - m_0 u^1 \Big) + 2 u^1 u^1_x + u^2 u^2_x + q u^1_x, \\
u^2_t & = q_{xxx} \Big(- \frac{m_2}{2 u^2} - \frac{m_0}{2} u^2\Big) + u^2 u^1_x + q u^2_x, \\
u^1 & = \frac{m_2}{2} q_{xx} + q.
\end{aligned}
$$
Differentiating the last equation and rearranging terms, we get $- \frac{m_2}{2}q_{xxx} = q_x - u^1_x$. Substituting it into the first, we get the equivalent form of the previous PDE with constraint
$$
\begin{aligned}
u^1_t & = \frac{m_1}{2} q_{xxx} + 2 u^1 q_x + q u^1_x + u^2u^2_x, \\
u^2_t & = - \frac{m_0}{4} \frac{q_{xxx}}{u^2} + (qu^2)_x, \\
u^1 & = q + \frac{m_2}{2} q_{xx}.
\end{aligned}
$$
This is a 3-parameter family of integrable systems. For $m_2 = 0$ we obtain the two-component Dullin-Gottwald-Holm equation \cite[(3)]{gdh2}. If, in addition,  $m_1 = 0$, then we get two-component Camassa-Holm equation \cite[(3) and (4)]{ch2}. 
}
\end{Ex}

\begin{Ex}  \label{eq:3.5} 
\rm{
We actually can generalise the equations of type I and III to the case of $k$ constraints, $k > 1$. Consider the expansion \eqref{eq:wsinfty} up to $\varepsilon^2$ :
$$
\begin{aligned}
\bar{\mathrm w} (\varepsilon) & = 1 + \varepsilon \mathcal w_{0, \epsilon} + \varepsilon^2 \mathcal w_{1, \infty} + \dots.
\end{aligned}
$$
If $\bar m(0) \neq 0$, then $\mathcal w_{0, \infty}$ and $\mathcal w_{1, \infty}$ are formal differential series. Now substitute the decompositions for $\bar{\mathrm w} (\varepsilon), \bar \sigma(\varepsilon)$ and $\bar m(\varepsilon)$ into the Gelfand-Dikii identity \eqref{r3}. Renaming $q^1 = \mathcal w_{0, \infty}$, $q^2=\mathcal w_{1, \infty} $ we get
\begin{equation*}
    \begin{aligned}
    0 & = \bar m(\varepsilon) \left(\bar{\mathrm w}_{xx} (\varepsilon) \bar{\mathrm w} (\varepsilon) - \frac{1}{2} (\bar{\mathrm w}_{x} (\varepsilon))^2\right) + \bar \sigma(\varepsilon) (\bar{\mathrm w}(\varepsilon))^2 - 1 = \varepsilon \Big( \bar \sigma_1 + 2 q^1 + \bar{m_0} q^1_{xx} \Big) + \\
    & + \varepsilon^2  \Big(\bar \sigma_2 + 2 q^2 + \bar m_0 q^2_{xx} + \bar m_1 q^1_{xx} + \bar m_0 q^1 q^1_{xx} + 2 \bar \sigma_1 q^1 + (q_1)^2 - \frac{1}{2} \bar m_0 (q^1_x)^2\Big) + \dots.
    \end{aligned}
\end{equation*}
This yields differential relations for $q^1, q^2$. Differential operator $\operatorname{Id} + D^2$ is formally invertible, so these constraints imply that $q^1$ is a differential series in $\bar \sigma_1$ and its derivatives and $q^2$ is a differential series in $\bar \sigma_1, \bar \sigma_2$ and their derivatives.

Now consider the expansion \eqref{eq:xislambda11} up to $\varepsilon^2$:
$$
\begin{aligned}
\bar{\xi}(\varepsilon) & = \xi_{-1, \infty} + \varepsilon \xi_{0, \infty} + \varepsilon^2 \xi_{1, \infty} + \dots = \\
& = u_x + \varepsilon \Big( q^1_{xxx} \zeta + (L + q^1 \operatorname{Id})u_x\Big) + \varepsilon^2 \Big( q^2_{xxx} \zeta + q^1_{xxx} L\zeta + (L^2 + q^1 L + q^2 \operatorname{Id})u_x\Big) + \dots. \\
\end{aligned}
$$
Taking $\xi_{1, \infty}$ we get a PDE with two differential constraints
$$
\begin{aligned}
u_t & = q^2_{xxx} \zeta + q^1_{xxx} L\zeta + (L^2 + q^1 L + q^2 \operatorname{Id})u_x, \\
0 & = \bar \sigma_1 + 2 q^1 + \bar m_0 q^1_{xx}, \\
0 & =  \bar \sigma_2 + 2 q^2 + \bar m_0 q^2_{xx} + \bar m_1 q^1_{xx} + \bar m_0 q^1 q^1_{xx} + 2 \bar \sigma_1 q^1 + (q_1)^2 - \frac{1}{2} \bar m_0 (q^1_x)^2.
\end{aligned}
$$
For $L$ differentially non-degenerate, $N = 0$ and $m(t) = t^n$ this yields the general form of Camassa-Holm equation CH($n$,$2$) from \cite{ih}. Taking the expansion up to a higher order, one obtains a greater number of differential constraints. }
\end{Ex}

\begin{Ex}{\rm
In  the previous examples, we have shown that many notable integrable systems are special cases of the systems from Section 2. We now describe  one of the simplest {\it new} examples.  
By construction, integrable systems we deal with are written in invariant form that is independent on the choice of a local coordinate chart. In particular, in order to make our system more symmetric, we may choose local coordinates related to the roots of the polynomial $m(\lambda)$.

 The next 3-component example is build starting with  $\sigma(\lambda) = \det (L - \lambda\Id)$ and  $m(\lambda)= (\lambda_1-\lambda)(\lambda_2-\lambda)(\lambda_3- \lambda)$ and choosing local coordinates  $(u^1,u^2, u^3)$ to be
$$
u^i = \det (\lambda_i \Id - L) \prod_{s\ne i} \frac{1}{\lambda_i - \lambda_s}.
$$  
In particular,   $u^i$ is proportional to $\sigma(\lambda_i)$ with some constant factor which is not essential.   Moreover,  $\zeta = -(u^1, u^2, u^3)^\top$ and $(L-\lambda_i\Id)^{-1} \zeta = e_i$.  Now if we take an arbitrary linear combination of the (commuting) evolutionary PDEs \eqref{eq:t2} of Type II, we get the following integrable system:

\begin{equation} \label{eq:pd_bols} 
 \begin{pmatrix} u_t^1 \\ u_t^2 \\ u_t^3 \end{pmatrix}   = \begin{pmatrix} c_1\left({1}/{\sqrt{u^1}}\right)_{xxx} \\ c_2\left({1}/{\sqrt{u^2}}\right)_{xxx} \\ c_3 \left({1}/{\sqrt{u^3}}\right)_{xxx} \end{pmatrix}  + A(u) \begin{pmatrix} u^1_x \\   u^2_x \\   u^3_x\end{pmatrix} .\end{equation}
where $c_1, c_2, c_3$ are arbitrary constants and $A(u)$ is the $3\times 3$ matrix with the components 
$$
A^j_i  = \frac{u^j}{\lambda_i-\lambda_j}\left( \frac{c_i}{ (u^i)^{3/2}} - \frac{c_j }{ (u^j)^{3/2}}       \right), \quad  j\ne i, \quad \mbox{and} \quad
A^i_i =    - \frac{c_i}{(u^i)^{3/2}}   - \sum_{j\ne i}  A^j_i.    
$$
One can also write it as follows:
$$
u^j_t = c_j \left(\left(\frac{1}{\sqrt{u^j}}\right)_{xxx} - \frac{u^j_x}{(u^j)^{3/2}}\right) + 
\sum_{i\ne j} \frac{u^ju^i_x - u^i u^j_x}{\lambda_i-\lambda_j}\left( \frac{c_i}{ u_i^{3/2}} - \frac{c_j }{ u_j^{3/2}}       \right).
$$
Here the first term represents the system of  three {\it uncoupled}  Harry Dym type equations (see Example \ref{kdv}), but the second term mixes all the variables.

The recursion formula for the conservation laws gives the following explicit  formula for the first six of them.
 The first two  corresponding to the root $\lambda_1$ are  
 \begin{equation} \label{eq:zero}  \sqrt{u^1} \quad \mbox{and} \quad \frac{   
\left( {u^1} u^1_{xx}
-\frac{5}{4}  \left(u^1_x\right)^{2}\right)
-2 {\left(u^1\right)}^{2} 
\left(  
\frac{u^1 + u^2}{\lambda_1-\lambda_2} 
+
\frac{u^1 + u^3}{\lambda_1 - \lambda_3} 
 - 1 \right)}
{ (u^1)^{3/2}}.\end{equation}  
The other four correspond to the roots  $\lambda_2$, $\lambda_3$ and can be obtained from the above formulas by cyclic permutation of  indices 1,2,3. 

One can also find, using the procedure described in Section 2, 
 the formulas for  commuting flows. Actually, the commuting flows of the lowest order are \eqref{eq:pd_bols}  with arbitrarily chosen $c_1, c_2, c_3$. Notice that this example  can be naturally generalised to the case of an arbitrary number of components.

}\end{Ex}


\section{Proofs of Theorems \ref{t1}  and  \ref{t2}} \label{sect3}

We start with  the following Lemma.

\begin{Lemma} \label{lm2}
Under the assumptions of Theorem \ref{t1}, we have
$$
(L^*_\lambda)^{-1} \ddd \sigma(\mu) = \frac{1}{\mu - \lambda} \Bigg( \ddd \sigma(\mu) - \frac{\sigma(\mu)}{\sigma(\lambda)} \ddd \sigma(\lambda) \Bigg).
$$
\end{Lemma}
\begin{proof}
Condition \eqref{eq:05} reads
\begin{equation}\label{s1}
    L^*\ddd \sigma(\mu) = \sigma(\mu) \ddd f + \mu \ddd \sigma(\mu) = \sigma(\mu) \ddd f + (\mu - \lambda)\ddd \sigma(\mu) + \lambda \ddd \sigma(\mu).
\end{equation}
Recall that by construction
$(L_\lambda^*)^{-1} \ddd f = \frac{1}{\sigma(\lambda)} \ddd \sigma(\lambda).$
Rearranging the terms and multiplying both sides of \eqref{s1} by $(L^*_{\lambda})^{-1}$, we get the statement of Lemma.
\end{proof}

Now let us recall some basic formulas and introduce some notations. We denote the derivative coordinates of order $j$ by $u^\alpha_{x^j}$ and set $u^\alpha_{x^0} = u^\alpha$. Consider a formal evolutionary vector field  $\xi$ with components $\xi^i$. The (Lie) derivative of a formal differential series $\mathrm w$ along $\xi$ is defined by
$$
\mathcal L_\xi \mathrm w = \sum \limits_{j = 0}^\infty \pd{\mathrm w}{u^\alpha_{x^j}} D^j(\xi^\alpha)
$$
with summation over $\alpha$ assumed, $\alpha = 1,\dots,n$. 

Let $\xi = \xi(\lambda)$ be the formal vector field defined by \eqref{eq:xiformal} and associated with the PDEs from Theorem \ref{t1}. We have:
\begin{equation*}
    \begin{aligned}
    & \mathcal L_{\xi(\lambda)} \sigma(\mu) = \pd{\sigma(\mu)}{u^\alpha} \xi^\alpha(\lambda) = \pd{\sigma(\mu)}{u^\alpha} \Big(\mathrm w_{xxx}(\lambda) \big(L - \lambda \operatorname{Id} \big)^{-1} \zeta + \mathrm w(\lambda ) \big(L - \lambda \operatorname{Id}\big)^{-1} u_x\Big)^\alpha = \\
    & = \mathrm w_{xxx}(\lambda) \pd{\sigma^\mu}{u^\alpha} \Big(L^{-1}_\lambda\Big)^\alpha_q \zeta^q + \mathrm w (\lambda) \pd{\sigma^\mu}{u^\alpha} \Big(L^{-1}_\lambda\Big)^\alpha_q u^q_x = \frac{1}{\mu - \lambda}\mathrm w_{xxx} (\lambda) \Big( m(\mu) - \frac{\sigma(\mu)}{\sigma(\lambda)} m(\lambda)\Big) + \\
    & + \frac{1}{\mu - \lambda}\mathrm w(\lambda) \Big( \sigma_x(\mu) - \frac{\sigma(\mu)}{\sigma(\lambda)}\sigma_x (\lambda)\Big).
    \end{aligned}
\end{equation*}
The last step follows from Lemma \ref{lm2}. Further rearranging terms and using  \eqref{id}, we get
\begin{equation}\label{r2}
    \begin{aligned}
    & \mathcal L_{\xi(\lambda)}\sigma(\mu) = \\
    & = \frac{1}{\mu - \lambda} \Bigg(m(\mu) \mathrm w_{xxx} (\lambda) + w(\lambda) \sigma_x(\mu) - \frac{\sigma(\mu)}{\sigma(\lambda)} \Big( m(\lambda) \mathrm w_{xxx} (\lambda) + \mathrm w (\lambda) \sigma_x (\lambda)\Big)\Bigg) = \\
    & = \frac{1}{\mu - \lambda} \Bigg( m(\mu) \mathrm w_{xxx} (\lambda) + w(\lambda) \sigma_x(\mu) - \frac{\sigma(\mu)}{\sigma(\lambda)} \Big( - 2 \sigma(\lambda) \mathrm w_x (\lambda) \Big) \Bigg) = \\
    & = \frac{1}{\mu - \lambda} \Bigg( m(\mu) \mathrm w_{xxx}(\lambda) + 2 \sigma(\mu) \mathrm w_x(\lambda) + \sigma_x(\mu) \mathrm w(\lambda) \Bigg).
    \end{aligned}
\end{equation}

Now let us proceed with the proof. Consider a pair of formal differential series $\mathrm w_1, \mathrm w_2$. We use notation $\mathrm w_1 \sim \mathrm w_2$, if there exists a formal differential series $\mathrm u$, such that $\mathrm w_1 - \mathrm w_2 = D \mathrm u$. In particular, the Leibnitz rule for $D$ implies, that $\mathrm w_1 D (\mathrm w_2) \sim - D(\mathrm w_1) \mathrm w_2$. More generally,  it yields the formula
$$
\mathrm w_1 D^j (\mathrm w_2) \sim (-1)^j D^j(\mathrm w_1) \mathrm w_2.
$$
Using the chain rule $\vpd{\mathrm v(\mu)}{u^\alpha} = \mathrm w(\mu) \pd{\sigma(\mu)}{u^\alpha}$, we get the following sequence of relations
\begin{equation*}
    \begin{aligned}
    & \mathcal L_{\xi(\lambda)} \mathrm v(\mu) = \sum \limits_{j = 0}^\infty \pd{\mathrm v(\mu)}{u^\alpha_{x^j}} D^j (\xi^\alpha(\lambda)) \sim \sum \limits_{j = 0}^\infty (-1)^j D^j \Big(\pd{\mathrm v(\mu)}{u^\alpha_{x^j}}\Big) \xi^\alpha(\lambda) = \vpd{\mathrm v(\mu)}{u^\alpha} \xi^\alpha(\lambda) = \\
    & = \mathrm w (\mu) \pd{\sigma(\mu)}{u^\alpha} \xi^\alpha (\lambda)  = \mathrm w(\mu) \mathcal L_{\xi(\lambda)} \sigma(\mu) = \frac{1}{\mu - \lambda} m(\mu) \mathrm w(\mu) \mathrm w_{xxx} (\lambda) + \frac{2}{\mu - \lambda} \mathrm w(\mu) \mathrm w_x (\lambda) \sigma(\mu) + \\
    & + \frac{1}{\mu - \lambda} \mathrm w(\mu) \mathrm w (\lambda) \sigma_x(\mu) \sim - \frac{1}{\mu - \lambda} m(\mu) \mathrm w_{xxx} (\mu) \mathrm w(\lambda) - \frac{2}{\mu - \lambda} \big(\mathrm w(\mu) \sigma(\mu)\big)_x \mathrm w(\lambda) + \\
    &+ \frac{1}{\mu - \lambda} \mathrm w(\mu) \sigma_x (\mu) \mathrm w(\lambda) = - \frac{1}{\mu - \lambda} \Bigg( m(\mu) \mathrm w_{xxx}(\lambda) + 2 \sigma(\mu) \mathrm w_x(\lambda) + \sigma_x(\mu) \mathrm w(\lambda) \Bigg) \mathrm w(\lambda) = 0\\
    \end{aligned}
\end{equation*}
Here we used formula \eqref{r2} and, again, identity \eqref{id}. Thus, we get $\mathcal L_{\xi(\lambda)} \mathrm v(\mu) \sim 0$, meaning that $\mathrm v (\mu)$ is a formal conservation law for the flow $\xi(\lambda)$. For $\lambda = \infty$ the proof is essentially the same.

Now let us proceed to the commuting flows. We will need the following Lemma.

\begin{Lemma}\label{fom}
Assume that $\sigma(\mu)$ and $\mathrm w(\mu)$ are related by Gelfand-Dikii identity \eqref{r3} and $\mathcal L_{\xi(\lambda)}\sigma(\mu)$ is given by \eqref{r2}. Then 
\begin{equation}\label{r1}
\mathcal L_{\xi(\lambda)} \mathrm w(\mu) = \frac{1}{\mu - \lambda}\Big(\mathrm w_x(\mu) \mathrm w(\lambda) - \mathrm w(\mu) \mathrm w_x(\lambda)\Big).    
\end{equation}
\end{Lemma}
\begin{proof}
We start with applying $\xi(\lambda)$ to Gelfand-Dikii identity \eqref{r3} and multiplying the result by $\mathrm w(\mu)$  (we also take into account the fact that $\mathcal L_\xi$ commute with $D$):
\begin{equation}\label{ss}
\begin{aligned}
0 & = \mathrm w(\mu) \mathcal L_{\xi(\lambda)}\Bigg( m(\mu) \left(\mathrm w_{xx} (\mu) \mathrm w (\mu) - \frac{1}{2} (\mathrm w_{x} (\mu))^2\right) + \sigma(\mu) \mathrm w^2(\mu)\Bigg) = \\
    & = m(\mu) \Big( \mathrm w^2(\mu) \mathcal L_{\xi(\lambda)}\mathrm w_{xx}(\mu) +  \mathrm w_{xx}(\mu) \mathrm w(\mu)\mathcal L_{\xi(\lambda)}\mathrm w(\mu)  - \mathrm w_{x}(\mu) \mathrm w(\mu) \mathcal L_{\xi(\lambda)}\mathrm w_{x}(\mu)\Big) + \\
    &  + 2 \sigma(\mu) \mathrm w^2(\mu) \mathcal L_{\xi(\lambda)}\mathrm w(\mu) +  \mathrm w^3(\mu)\mathcal L_{\xi(\lambda)} \sigma(\mu)= \\
& = \left(   m(\mu) \Big( \mathrm w^2(\mu) D^2   -  \mathrm w_{x}(\mu) \mathrm w(\mu) D  +  \mathrm w_{xx}(\mu) \mathrm w(\mu)\Big) + 2 \sigma(\mu) \mathrm w^2(\mu) \right)\mathcal L_{\xi(\lambda)}\mathrm w(\mu)  + \\
&    +  \mathrm w^3(\mu)\mathcal L_{\xi(\lambda)} \sigma(\mu) = \cal R \big(\mathcal L_{\xi(\lambda)}\mathrm w(\mu) \big)+  \mathrm w^3(\mu)\mathcal L_{\xi(\lambda)} \sigma(\mu) 
\end{aligned}  
\end{equation}
with
$$
\begin{aligned}
\mathcal R &=    m(\mu) \Big( \mathrm w^2(\mu) D^2   -  \mathrm w_{x}(\mu) \mathrm w(\mu) D  +  \mathrm w_{xx}(\mu) \mathrm w(\mu) \Id \Big) + 2 \sigma(\mu) \mathrm w^2(\mu)\Id   = \\
&= \big(2 - m(\mu) \mathrm w(\mu) \mathrm w_{xx} (\mu) + m(\mu) \mathrm w^2_x(\mu)\big) \operatorname{Id} - m(\mu) \mathrm w(\mu) \mathrm w_x(\mu) D + m(\mu) \mathrm w^2(\mu) D^2,
\end{aligned}
$$
 where in the latter relation we substitute $\sigma(\mu) \mathrm w^2(\mu) = 1 -
 m(\mu) \left(\mathrm w_{xx} (\mu) \mathrm w (\mu) - \frac{1}{2} (\mathrm w_{x} (\mu))^2\right)$ from \eqref{r3}.
 Thus, we have the identity
\begin{equation}
\label{eq:Rid}
\mathcal R \big(\mathcal L_{\xi(\lambda)} \mathrm w (\mu)\big) = - \mathrm w^3(\mu) \mathcal L_{\xi(\lambda)}\sigma(\mu) .
\end{equation}

Note that $\cal R$ is a (formally) invertible differential operator.  Therefore, it suffices to verify that 
$\mathcal L_{\xi(\lambda)} \mathrm w (\mu)$ defined by \eqref{r1} satisfies \eqref{eq:Rid} or, equivalently,
$$
\cal R \big(\mathrm w_x(\mu) \mathrm w(\lambda) - \mathrm w(\mu) \mathrm w_x(\lambda)\big) + (\mu -\lambda)\mathrm w^3(\mu) \mathcal L_{\xi(\lambda)}\sigma(\mu) =0.
$$ 

Computing the l.h.s. of this relation gives:
\begin{equation*}
    \begin{aligned}
           & \ \ \ \, 
           m(\mu)\mathrm w^2(\mu) 
              \Big( \mathrm w_x(\mu)  \mathrm w(\lambda) -  \mathrm w(\mu)  \mathrm w_x(\lambda)\Big)_{xx}             
        + m(\mu)\mathrm w_{xx} (\mu) \mathrm w(\mu) 
             \Big( \mathrm w_x(\mu)  \mathrm w(\lambda) -  \mathrm w(\mu)  \mathrm w_x(\lambda)\Big) - \\
     & - m(\mu)\mathrm w_x(\mu) \mathrm w(\mu) 
             \Big( \mathrm w_x(\mu)  \mathrm w(\lambda) -   \mathrm w(\mu) \mathrm w_x(\lambda)\Big)_x + 
         2 \sigma(\mu)  \mathrm w^2(\mu) 
              \Big( \mathrm w_x(\mu)  \mathrm w(\lambda) -  \mathrm w(\mu)  \mathrm w_x(\lambda) \Big) + \\
    & + (\mu - \lambda) \mathrm w^3(\mu) \mathcal L_{\xi(\lambda)}\sigma(\mu) = \\
    & = m(\mu)\mathrm w^2(\mu)\Big( \mathrm w_{xxx} (\mu)    \mathrm w(\lambda)  -  \mathrm w(\mu)   \mathrm w_{xxx}(\lambda)\Big) + 
    2 \sigma(\mu)  \mathrm w^2(\mu) \Big( \mathrm w_x(\mu)  \mathrm w(\lambda) -  \mathrm w(\mu)  \mathrm w_x(\lambda)\Big) \\  &+  (\mu - \lambda) \mathrm w^3(\mu)\mathcal L_{\xi(\lambda)}\sigma(\mu).  
    \end{aligned}
\end{equation*}

Adding and subtracting $\sigma_x(\mu) \mathrm w (\lambda) \mathrm w^3(\mu)$ we arrive to the identity
$$
\begin{aligned}
\Big( (\mu -\lambda) \mathcal L_{\xi(\lambda)}\sigma (\mu) - & m(\mu) \mathrm w_{xxx} (\lambda) - 2 \sigma(\mu) \mathrm w_x (\lambda) - \sigma_x(\mu) \mathrm w(\lambda) \Big)\mathrm w^3(\mu) + \\
+ &\Big( m(\mu) \mathrm w_{xxx}(\mu) + 2 \sigma(\mu) \mathrm w_x(\mu) + \sigma_x(\mu) \mathrm w(\mu) \Big) \mathrm w^2(\mu) \mathrm w(\lambda) = 0,
\end{aligned}
$$
where the first term vanishes due to \eqref{r2} and the second due to \eqref{id}, completing the proof.
\end{proof}

As we deal with evolutionary vector fields, it is enough to check that $\xi(\lambda)$ and $\xi(\nu)$ commute, acting on coordinate functions. Fix three pairwise distinct $\lambda, \mu, \nu$. From Lemma \ref{fom}, we get
\begin{equation*}
    \begin{aligned}
    & \mathcal L_{\xi(\nu)} \mathcal L_{\xi(\lambda)} \sigma(\mu) = \frac{1}{\mu - \lambda}  \mathcal L_{\xi(\nu)} \Big(m(\mu) \mathrm w_{xxx}(\lambda) + 2 \sigma(\mu) \mathrm w_x(\lambda) + \sigma_x(\mu) \mathrm w(\lambda) \Big) = \\
    & = \frac{m(\mu)}{(\mu - \lambda)(\lambda - \nu)} \Big(\mathrm w_{xxxx}(\mu) \mathrm w(\lambda) + 2 \mathrm w_{xxx}(\mu) \mathrm w_x(\lambda) - \mathrm w(\mu) \mathrm w_{xxxx}(\lambda) - 2 \mathrm w_x(\mu) \mathrm w_{xxx}(\lambda)\Big) + \\
    & + \frac{2}{(\mu - \lambda)(\mu - \nu)} \Big(m(\mu) \mathrm w_{xxx}(\nu) + 2 \sigma(\mu) \mathrm w_x(\nu) + \sigma_x(\mu) \mathrm w(\nu)\Big)\mathrm w_x(\lambda) + \\
    & + \frac{1}{(\mu - \lambda)(\lambda - \nu)} \Big(2 \sigma(\mu) \Big(\mathrm w_{xx}(\mu) \mathrm w(\nu) - \mathrm w(\mu) \mathrm w_{xx}(\nu)\Big) +  \sigma_x(\mu) \Big(\mathrm w_x(\mu) \mathrm w(\nu) - \mathrm w(\mu) \mathrm w_x(\nu)\Big) \Big) + \\
    & + \frac{1}{(\mu - \lambda)(\mu - \nu)} \Big(m(\mu) \mathrm w_{xxxx}(\nu) + 2 \sigma_x(\mu) \mathrm w_x(\nu) + 2 \sigma(\mu) \mathrm w_{xx}(\nu) + \sigma_{xx}(\mu) \mathrm w(\nu) + \sigma_{x}(\mu) \mathrm w_x(\nu)\Big) \mathrm w(\lambda)
    \end{aligned}
\end{equation*}
The identity
$$
\frac{1}{(\mu - \lambda)(\mu - \nu)} - \frac{1}{(\mu - \lambda)(\lambda - \nu)} + \frac{1}{(\mu - \nu)(\lambda - \nu)} =  0.
$$
implies that the r.h.s. of formula for $\mathcal L_{ \xi(\nu)}\mathcal L_{ \xi(\lambda)}\sigma(\mu)$ is symmetric in $\lambda, \nu$. Thus, $\mathcal L_{\xi(\nu)}$ and  $\mathcal L_{\xi(\lambda)}$  commute on $\sigma(\mu)$. 

Now recall that $\sigma(\mu)$ is constructed from a generic conservation law $f$. This implies that in the expansion 
$$
\sigma(\mu) = \sigma_0 + \mu \sigma_1 + \dots
$$
the differentials of $\sigma_0, \dots, \sigma_{n - 1}$ are linearly independent almost everywhere. Thus, one can take them as coordinates $u^i = \sigma_{i - 1}$ and in these coordinates $(\mathcal L_{\xi(\nu)}\mathcal L_{\xi(\lambda)} - \mathcal L_{\xi(\lambda)}\mathcal L_{\xi(\nu)})  u^i = 0$, as required.  This completes the proof of Theorem \ref{t1}.

To verify the statement of Theorem \ref{t2}, we first need to show that the coefficients $\xi_{s,\lambda_i}$, $\mathrm v_{s,\lambda_i}$ and $\mathrm w_{s,\lambda_i}$ of the $\varepsilon$-expansions \eqref{eq:expandxi} and \eqref{eq:expandvw}  
are well defined and are differential polynomials in $u^1,\dots, u^n$.   Indeed,  by definition,
\begin{equation}
\label{eq:proofth3}
\mathrm v (\lambda) = \mathrm v (\sigma(\lambda), m(\lambda)) = \sqrt{2} \sum_{s=0}^\infty \bigl(m(\lambda)\bigr)^s \mathcal v_s(\lambda),
\end{equation}
where $\mathcal v_s (\lambda)$ is a differential polynomial obtained from the homogeneous differential polynomial  $\mathcal u_{2s+1}(\sigma, \sigma_x, \sigma_{xx}, \dots)$ of degree $2s$ by substitution $\sigma = \sigma(\lambda; u)$.

We are interested in the expansion of $\mathrm v (\lambda_i + \varepsilon) = \sum \mathrm v_{s,\lambda_i} \varepsilon^s$ under the assumption that $m(\lambda_i) = 0$.  Since $m(\lambda_i + \varepsilon)= a_1 \varepsilon + a_2 \varepsilon^2 + \dots$ and, therefore, $\bigl( m(\lambda_i + \varepsilon)\bigr)^s = a_1^s \varepsilon^s +\dots$, we see from  expansion \eqref{eq:proofth3} that 
$\mathrm v_{s,\lambda_i}$ is defined from the first $s+1$ coefficients  $\mathcal v_0,\dots, \mathcal v_s$.  Hence,  $v_{s,\lambda_i}$ is a non-homogeneous differential polynomial of degree at most $2s$. 

The proof for  $\mathrm w_{s,\lambda_i}$ is literally the same.   The conclusion for $\mathrm w_{s,\lambda_i}$  immediately follows from the explicit formula of $\xi(\lambda)$ in terms of $\mathrm w(\lambda)$, see \eqref{eq:main} and \eqref{eq:xiformal}.

The explicit form \eqref{eq:explform} of the flows  $u_{t_{\lambda_i}} = \xi (\lambda_i)$ for a root $\lambda_i\in\mathbb C$ of $m(\lambda)$  is straightforward.   Indeed, setting $m(\lambda_i)=0$ in \eqref{eq:main}  gives  $\sigma(\lambda_i) \bigr(\mathrm w(\lambda_i)\bigr)^2 - 1 = 0$, or equivalently, $\mathrm w(\lambda_i) =  \frac{1}{\sqrt{\sigma(\lambda_i)}}$.  Substituting this expression into the first equation of \eqref{eq:main} gives \eqref{eq:explform}, as required.   Similarly, for $\lambda_i = \infty$ we 
set $\bar m(0) = 0$  in \eqref{eq:formal3}  to get  $q=\frac{1}{2} f$,  which after substitution into the first equation of \eqref{eq:formal3} immediately gives \eqref{eq:explforminfty}, completing the proof of Theorem \ref{t2}.



\section{Proof of Theorem \ref{t3}}\label{sect5}

Since $L$ is differentially non-degenerate,    this operator is diagonalisable almost everywhere.    At ``non-diagonalisable'' points,  the conclusion of Theorem \ref{t3} can be derived by continuity arguments.  So w.l.o.g. we assume that $L= \diag(x^1,\dots,x^n)$.  
Recall that $f$  from equation \eqref{eq:05}
is a conservation law for $L$.    Then $f$ is a sum of $n$ 
functions such that the $i$th function depends on $x^i$ only.  Therefore, for every $i$ the function   
$f_i:= \tfrac{\partial f}{\partial x^i}$     depends on $x^i$ only. 

Next, we consider relation \eqref{eq:05}. In coordinates, it reads:
\begin{eqnarray*} 
\frac{\partial \ln(\sigma)}{\partial x^1} & = &  \frac{f_1(x^1) }{ x^1 - \lambda} \\
 &\vdots& \\
\frac{\partial \ln(\sigma)}{\partial x^n} & = &  \frac{f_n(x^n)}{x^n - \lambda}. 
\end{eqnarray*} 
Hence,  the system \eqref{eq:05} of  $n$  PDEs  is actually a system of $n$  ODEs in different  variables.
Its  solution must  be of the form $\sigma = c(\lambda) \cdot \sigma_1 \dots \sigma_n$ with
\begin{equation}
\label{eq:pi}  
  \sigma_i(\lambda, x^i) = \exp\left(\int_{s_i}^{x^i} \frac{f_i(s)}{s - \lambda} \ddd s\right) 
\end{equation} 
where  $c(\lambda)$ is an arbitrary function and $(s_1,...,s_n)$ is an arbitrary point; we assume that all $s_i\ne 0$.  

Next, consider relation \eqref{zeta2}. For our $\sigma(\lambda)=  c(\lambda) \cdot \sigma_1 \dots \sigma_n$  it reads
 \begin{equation}
 \label{eq:matv3}
	\left(C + \sum_{i=1}^n\zeta^i \frac{ f_i }{x^i- \lambda}\right) \, \sigma_1 \dots \sigma_n = \frac{m(\lambda)}{c(\lambda)}:=\widehat m(\lambda).
\end{equation}	

In the left hand side of this relation,  $f_i$ and $\zeta_i$ are smooth functions in $x$ which are independent on $\lambda$, whereas the r.h.s. is a function independent of $x$.  The following statement shows that under these conditions, $f_i$'s have to be constants and, moreover, very special.

\begin{Lemma} \label{prop:3} 
The functions $f_i$ are integer constants different from zero and no greater than $1$. Moreover, for every $i$ such that    $f_i\ne 1$ we have $\zeta_i=0$. 
\end{Lemma} 

\begin{proof} Integration by parts gives 
$$
\sigma_i(\lambda , x^i)= \exp\left(\int_{s_i}^{x^i} \frac{f_i(s)}{s -\lambda } \ddd s  \right) = 
\exp\left( f_i(x^i)\ln (x^i-\lambda)   - f_i(s_i)\ln (s_i-\lambda) - \int_{s_i}^{x^i} f'_i(s)  \ln(s-\lambda ) \ddd s\right)     
$$ 
implying $\sigma(x,\lambda )= (x^1-\lambda)^{f_1(x^1)} (x^2-\lambda)^{f_2(x^2)}\dots (x^n-\lambda)^{f_n(x^n)} \tilde \sigma(\lambda ,x)$ where the function $\tilde \sigma$ has neither zeros 
nor poles.  The equation \eqref{eq:matv3} reads then 
\begin{equation} 
\label{eq:5}  
\left(C + \sum_{s=1}^n \frac{\zeta^s f_s(x^s)}{x^s-\lambda} \right)  (x^1-\lambda)^{f_1(x^1)} (x^2-\lambda)^{f_2(x^2)}\dots (x^n-\lambda)^{f_n(x^n)} \tilde \sigma (\lambda , x) = \widehat m(\lambda). 
\end{equation} 
Note that  the function $ \int_{s_i}^{x^i} f'_i(s) \ln(s - \lambda) \ddd s$ is locally  bounded, so the function  $\tilde \sigma(\lambda , x)$  is bounded for small $x^s-\lambda$ and is not zero.   

 Assume for a certain $i$ that $f_i$ is not integer at a point $(x^1,...,x^n)$. Substituting $\lambda = x^i$ (for generic $x^i$) makes  $\widehat m(\lambda  )=0$ 
or $\widehat m(\lambda )=\infty$ leading to a contradiction.     Indeed, the  first factor of \eqref{eq:5} has  integer order of zeros  and poles, so to compensate it $f_i$ must be integer for each $x^i$ and therefore constant. Note that the case $\widehat m(\lambda ) \equiv \infty $ is not allowed since the left hand side  is finite almost everywhere. 
 
Thus, all $f_i$ are integer. Then the function $\tilde \sigma(\lambda, x) $ depends on $\lambda$ only.  Further we assume that it is equal to $1$ since we can ``hide'' it in $\widehat m(\lambda)$ (we keep the same notation).   
Our equation  \eqref{eq:5}    then becomes: 
\begin{equation} 
\label{eq:6}  
\left(C+ \sum_{s=1}^n \frac{ \zeta^s f_s}{x^s-\lambda} \right)  (x^1-\lambda)^{f_1} (x^2-\lambda)^{f_2}\dots (x^n-\lambda)^{f_n}  = \widehat m(\lambda ),  \end{equation} 
where $f_i$ are  integer constants. 
 If  $\zeta^i\ne 0$ then the first factor of \eqref{eq:6} has a pole of order $1$  implying    $f_i=1$.  If $\zeta^i= 0$ and  $f_i> 0$, then  \eqref{eq:6} has zero
 for $\lambda =x^i$ which again leads to contradiction. It remains to notice that $f_i\ne 0$ since 
  $\ddd f, L^*\ddd f,\dots ,L^{n-1}\ddd f$ are linearly independent by our assumption.    \end{proof}

W.l.o.g. we assume that $f_i=1$ for    $i=1,...,k$ and the other $f_i$
 are negative  integers; we set $\ell_{k+1}= -f_{k+1},\dots ,\ell_{n}= -f_{n}$.   The equation  \eqref{eq:6} reads then 
\begin{equation} 
\label{eq:7}  
\left(C+ \sum_{s=1}^k \frac{\zeta^s }{x^s-\lambda} \right) =  
\frac{(x^{k+1}-\lambda)^{\ell_{k+1}}\dots  (x^{n}-\lambda)^{\ell_n}   }{(x^1-\lambda)(x^2-\lambda)\dots(x^k-\lambda)} \, \widehat m(\lambda).   
\end{equation} 
Notice that the expression in the l.h.s. can be written as a rational function in $\lambda$ of the form 
$\frac{P(\lambda)}{\prod (x^s - \lambda)}$, where $P(\lambda)=C (-\lambda)^k + \dots$ is a polynomial of degree at most $k$.   Similarly, the r.h.s. is $\frac{F(\lambda)}{\prod (x^s - \lambda)}$ where 
$F(\lambda) = (x^{k+1}-\lambda)^{\ell_{k+1}}\dots (x^{n}-\lambda)^{\ell_n}  \widehat m(\lambda)$.  Since $P(\lambda)=F(\lambda)$ we conclude that $\widehat m (\lambda) = m_d \lambda^d + \dots$ is a polynomial of degree at most $d= k -\ell_{k+1} - \dots - \ell_n$. In particular,  $d\ge 0$ and $m_d = (-1)^d C$.

Finally, it remains to notice that $\zeta^i$ can be found from \eqref{eq:7}  by using the partial fraction decomposition theorem which gives:
\begin{equation}
\label{eq:bols8}
\zeta^i = \frac{(x^{k+1} - x^i)^{\ell_{k+1}} \dots  (x^{n} - x^i)^{\ell_{n}} \widehat m(x^i)     }{\prod_{s=1, s\ne i}^k (x^s-x^i)}, \qquad  i = 1,\dots, k.
\end{equation}

Summarising this discussion we conclude that for $L=\diag(x^1,\dots, x^n)$,  we have  (up to scaling with a factor $c(\lambda)$)
$$
\begin{aligned}
f(x) &=\sum_{i=1}^k x^i - \sum_{i=k+1}^n \ell_i x^i \\
\sigma(\lambda, x) &=  \frac{\prod_{i=1}^k (x_i -\lambda)}{ \prod_{i=k+1}^n  (x_i - \lambda)^{\ell_i}}
\end{aligned}
$$
and $\zeta$ given by \eqref{eq:bols8}. 

To get the conclusion of Theorem \ref{t3} in its final form, we only need to combine the coordinates $x^{i}$'s,  $i=k+1,\dots, n$ into  groups depending on the values of  the exponents $\ell_i$  ($x^i$ and $x^j$ go to the same group iff  $\ell_i=\ell_j$).   After appropriate renumbering, we come to the desired description of $L$,  $\sigma$, $m$ and $\zeta$.

\section{Conclusion} 

In our paper, we  constructed new explicit families of integrable multi-component  evolutionary  equations with and without differential constraints, see Section \ref{sect1.1}.  The equations $u_{t_\lambda}= \xi (\lambda; u, u_x, u_{xx},\dots)$ within each family are parametrised by $\lambda\in \bar{\mathbb C}=\mathbb C \cup\{\infty\}$.  The corresponding (formal) evolutionary flows pairwise commute for all values of parameters and admit a family of common (formal) conservation laws $\mathrm v(\mu; u, u_x, u_{xx},\dots)$ also parametrised by $\mu\in \bar{\mathbb C}$.  For specific values of  $\lambda$  (namely, roots $\lambda_1,\dots, \lambda_d$ of a certain polynomial)  the above equation generates an hierarchy of usual (i.e. non-formal) commuting flows $\xi_{s,\lambda_i}$, $s=0,1,\dots$, defined by means of a differential polynomial of degree $2s+3$.  The equation $u_{t_{\lambda_i}} = \xi (\lambda_i; u, u_x, u_{xx}, u_{xxx})$,  initial term of this hierarchy, is a third order PDE system given by an elegant explicit formula.  All the subsequent terms can be found by means of explicit recurrent formulas.  Similar for conservation laws:   $\mathrm v(\lambda_i)$ generates an hierarchy of common polynomial conservation laws  $\mathrm v_{s,\lambda_i}$, $s=0,1,\dots$, for all the flows,  where $\mathrm v_{s,\lambda_i}$ is a differential polynomial of degree $2s$ that can be found explicitly by an iterative procedure.
These  families of integrable equations, for a simple choice of parameters, include and generalise many known examples of integrable systems. Some of multi-component  evolutionary  equations we constructed are essentially  new and they have no low-component analogues.

The construction is based on a new approach, which is rather differential-geometric than algebraic (in contrast to  many other  constructions of integrable systems which are often based on algebraic or algebraic-geometric concepts). Our results have been naturally obtained within the Nijenhuis Geometry programme initiated in \cite{nij}.  This suggests that Nijenhuis Geometry might be a convenient framework for studying further
properties of the constructed systems and generalizing them.   Because of its differential-geometric nature, our  constructions are invariant with respect to  the choice of coordinates on $\mathrm  M^n$; that is, the systems behave covariantly if we 
change unknown functions $u$  by a diffeomorphism $u_{\mathsf{new}}= u_{\mathsf{new}}(u_{\mathsf{old}})$. One can use this fact in the search of applications of our systems in natural sciences. Examples discussed in Section \ref{sect2.3} actually 
suggest that `physically relevant' variables correspond to those coordinates on $\mathrm  M^n$ in which the Nijenhuis operator has a nice form, e.g., the   `companion' and `diagonal' forms from \eqref{dnd}, or the form in which the components of $L$ are linear in coordinates. 

The famous integrable systems that we generalise to an arbitrary number of components (such as KdV, Camassa-Holm, Dullin-Gottwald-Holm, Harry Dym, Kaup-Boussinesq) have been intensively studied for decades; for these studies, a number of non-trivial geometric, algebraic and analytical methods were invented and successfully applied. The next natural step would be to figure out  how to adapt  these methods to new systems. In particular, it would be interesting to construct Lax representations for new systems, to find explicit solutions by the inverse scattering method, to construct a recursion operator  and, of course, to find physically relevant models that are described by new systems. We invite our fellow mathematicians and physicists to join this research.

\end{document}